\documentclass{article}
\usepackage{arxiv}

\usepackage{hyperref}
\usepackage{url}
\usepackage{algorithm}
\usepackage{algpseudocode}
\usepackage{graphicx}
\usepackage{amsmath,amsfonts,amsthm}
\usepackage{amssymb}
\usepackage{thmtools, thm-restate}
\usepackage{comment}
\usepackage{xcolor}
\usepackage{booktabs}
\usepackage{siunitx}

\newtheorem{lemma}{Lemma}

\newtheorem{remark}{Remark}
\newtheorem{corollary}{Corollary}
\newtheorem{theorem}{Theorem}
\newtheorem{proposition}{Proposition}
\newtheorem{definition}{Definition}
\newtheorem{assumption}{Assumption}

\usepackage[square]{natbib}
\usepackage[title]{appendix}
\usepackage{enumitem}

\newcommand{\x}{\mathbf{x}}

\usepackage{dsfont}
\newcommand{\indicator}[2]{\mathds{1}_{#1}\big({#2}\big)}

\newcommand{\Residue}{\mathcal{O}\Big(\frac{1}{\sqrt{N}}\Big)}

\newcommand{\empMu}[1]{\mathrm{Emp}_{\mu}(#1)}

\newcommand{\bfX}{\mathbf{X}}

\newcommand{\bfU}{\mathbf{U}}

\newcommand{\bfx}{\mathbf{x}}

\renewcommand{\t}{^{\mbox{\tiny\sf T}}}

\newcommand{\X}{\mathcal{X}}
\newcommand{\Y}{\mathcal{Y}}
\newcommand{\M}{\mathcal{M}}

\newcommand{\U}{\mathcal{U}}

\renewcommand{\P}{\mathcal{P}}
\renewcommand{\S}{\mathcal{S}}

\newcommand{\B}{\mathcal{B}}

\newcommand\norm[1]{\left\lVert#1\right\rVert}

\usepackage{pifont}
\newcommand{\abs}[1]{\left\lvert#1\right \rvert}

\newcommand{\expct}[2]{\mathbb{E}_{#2}\left[#1\right]}
\newcommand{\opt}{\mathrm{opt}}
\newcommand{\prop}{\mathrm{prop}}
\newcommand{\dtv}[1]{\mathrm{d}_{\mathrm{TV}} \big( #1 \big)}
    
\newcommand{\pt}[1]{\textcolor{red}{[Panos:] #1}}
\newcommand{\bhavs}[1]{\textcolor{blue}{[Bhavini:] #1}}
\newcommand{\sg}[1]{\textcolor{orange}{[Scott: #1]}}

\usepackage{fancyhdr}

\fancyhfoffset[L]{0cm} 
\fancyhfoffset[R]{0cm} 
\title{\LARGE \bf Robust Mean-Field Games with Risk Aversion and Bounded Rationality
}

\author{
	Bhavini Jeloka%
	\thanks{Bhavini Jeloka is a PhD student with the School of Aerospace Engineering, Georgia Institute of Technology, Atlanta, GA, USA. Email:
		{\tt\small bjeloka3@gatech.edu}}
	\qquad Yue Guan%
	\thanks{Yue Guan is a PhD student with the School of Aerospace Engineering, Georgia Institute of Technology, Atlanta, GA, USA. Email:
		{\tt\small yguan44@gatech.edu}}
 	\qquad Panagiotis Tsiotras%
 	\thanks{Panagiotis Tsiotras is the David \& Andrew Lewis Chair Professor with the School of Aerospace Engineering, Georgia Institute of Technology, Atlanta, GA, USA. Email: {\tt\small tsiotras@gatech.edu}}
}

\begin{document}
\maketitle

\setlength{\abovedisplayskip}{5pt}
\setlength{\belowdisplayskip}{5pt}

\begin{abstract}
Recent advances in mean-field game literature enable the reduction of large-scale multi-agent problems to tractable interactions between a representative agent and a population distribution. 
However, existing approaches typically assume a fixed initial population distribution and fully rational agents, limiting robustness under distributional uncertainty and cognitive constraints.
We address these limitations by introducing risk aversion with respect to the initial population distribution and by incorporating bounded rationality to model deviations from fully rational decision-making agents.
The combination of these two elements yields a new and more general equilibrium concept, which we term the mean-field risk-averse quantal response equilibrium (MF-RQE).
We establish existence results and prove convergence of fixed-point iteration and fictitious play to MF-RQE. 
Building on these insights, we develop a scalable reinforcement learning algorithm for scenarios with large state–action spaces. 
Numerical experiments demonstrate that MF-RQE policies achieve improved robustness relative to classical mean-field approaches that optimize expected cumulative rewards under a fixed initial distribution and are restricted to entropy-based regularizers.
\end{abstract}

\section{Introduction}
Multi-agent decision-making is central to many real-world applications, including traffic control, swarm systems, and multi-robot coordination~\citep{yu2023decentralized, yang2023partially, rizk2018decision, berman2009optimized}.
As the number of agents grows, decision-making complexity increases exponentially, leading to the well-known \emph{curse of dimensionality}, making scalability a fundamental challenge in multi-agent systems.
Mean-field theory offers a principled approach to address this challenge by approximating large-population interactions in the infinite-agent limit~\citep{huang2006large}.
This framework gives rise to mean-field games (MFGs) for non-cooperative settings~\citep{huang2006large, cui2021approximately, guan2022shaping}, mean-field control (MFC) for cooperative problems~\citep{saldi2023partially, sen2019mean}, and more recently mean-field team games (MFTGs) for mixed cooperative–competitive scenarios~\citep{guan2024zero, shao2024reinforcementlearningfinitespace, jeloka2025learninglargescalecompetitiveteam}.

In MFGs, a multi-agent interaction is approximated by a game between a representative agent and the population distribution (mean-field), yielding a mean-field equilibrium (MFE) that approximates Nash equilibria in large but finite populations.
The appeal of MFE lies in its decentralized structure: agents deploy identical open-loop policies with respect to the mean field, avoiding costly or unreliable real-time population feedback.
However, classical MFGs rely on strong assumptions, including a fixed initial population distribution, perfect rationality, and risk neutrality.
In practice, agents operate across varying initial distributions.
Due to the iterative nature of solving MFE, recomputing policies for each such condition is typically computationally infeasible.
These challenges introduce uncertainty in population-level statistics and naturally give rise to risk preferences over a \emph{set} of possible initial distributions--a behavior that is not captured by expected-reward maximization under a single initial condition.
Furthermore, perfect rationality is rarely realistic in practice, as agents exhibit errors, biases, and systematic deviations from optimal behavior~\citep{mckelvey1995quantal, prediction1}.
To address these limitations, both risk aversion and bounded rationality must be incorporated into the mean-field framework.

The recent work of~\citet{mazumdar2025tractable} introduced the risk-averse quantal response equilibrium (RQE) for finite-agent Markov games, but suffers from the curse of dimensionality.
To address the shortcomings associated with large population regimes, in this work we introduce a risk-averse MFG framework that explicitly accounts for uncertainty in initial population distributions and deviations from perfect rationality.
We formulate the problem as a convex risk-averse optimization problem and introduce a new solution concept, the \emph{mean-field risk-averse quantal response equilibrium} (MF-RQE).

Our main contributions are summarized as follows: 1) a novel MFG formulation incorporating both risk aversion over initial distributions and bounded rationality, leading to the MF-RQE; 2) fixed-point iteration and fictitious play algorithms to compute MF-RQE, including function-approximation variants for large state–action spaces; 3) existence and convergence guarantees; 4) extensive experiments demonstrating that MF-RQE policies exhibit substantially improved robustness compared to risk-neutral policies optimized for a single initial distribution.

\section{Related Work}
Risk aversion and bounded rationality have largely been studied in isolation—through risk-averse and quantal response equilibria, respectively—with growing adoption in MARL and learning theory; see~\cite{mazumdar2025tractable} and references therein for Markov games. 
We focus on related work that examines their integration with MFGs.

\textbf{Risk Analysis in MFGs.}
Risk-sensitive mean-field games have been studied in both continuous and discrete-time settings.
Continuous-time analyses~\citep{moon2014linear, tembine2013risk, tembine2015risk, perrin2020fictitious} typically rely on common-noise formulations and differ fundamentally from the finite-horizon, discrete-time framework considered here.
In discrete time, prior work~\citep{saldi2020approximate, cheng2023risk, bonnans2021discrete} often assumes Polish state–action spaces or linear feedback control, whereas our finite state–action formulation admits a standard MDP structure with explicit dynamics and does not require real-time mean-field feedback, enabling scalability with large state-action spaces.
%
%
While our approach is related to exponential utility-based risk-sensitive formulations~\citep{saldi2020approximate}, our framework—unlike theirs—accounts for uncertainty in the mean-field flow itself, induced by uncertainty in the initial distribution, yielding a \emph{log-sum-exp} aggregation over multiple initial distributions.
Finally, all existing risk-sensitive MFG frameworks predominantly rely on Nash equilibrium as the solution concept and do not account for deviations from perfect rationality.

\textbf{Bounded Rationality in MFGs.} 
In MFGs, bounded rationality
has primarily appeared through entropy regularization to compute approximate Nash equilibria~\citep{cui2021approximately, guan2022shaping}, and was explicitly interpreted as bounded rationality in~\citet{eich2025bounded}, where the authors analyzed a stochastically perturbed mean-field game and computed the resulting mean-field quantal response equilibria (QRE).
%
%
%
%
%
Under suitable assumptions, \citet{eich2025bounded} show that bounded rationality yields both more realistic agent behavior and computational advantages, as the resulting QRE can be computed via fixed-point iteration and, under specific noise assumptions, reduces to an entropy-regularized MFG.
Building on these insights, we model bounded rationality in MFGs using a \emph{general class of convex regularizers}, enabling richer behavioral representations beyond entropy regularization while retaining computational tractability.
Moreover, unlike \citet{eich2025bounded}, we explicitly incorporate \emph{risk aversion} into the boundedly rational mean-field framework, thereby capturing agents’ preferences for the initial population distribution.

\textbf{Impact of Initial Distributions.} 
Mean-field games under uncertainty in the initial distribution remain relatively underexplored, with limited theory and few design methodologies offering performance guarantees under risk aversion or modeling error.
\citet{jin2025initialerrortolerantdistributed} propose an initial-error-tolerant distributed mean-field control (IET-DMFC) framework, but it is limited to linear–quadratic models and does not incorporate risk preferences.
\citet{cui2023learning} empirically study sensitivity to initial conditions in the MFC regime using open-loop policies learned via Dec-POMFPPO, but do not provide theoretical robustness guarantees.

\section{Notation} 
We use uppercase letters to denote random variables (e.g., $X$ and $\M$) and lowercase letters to denote their realizations (e.g., $x$ and $\mu$).
We use $[N]$ to denote the discrete set of the natural numbers $\{1, 2, \ldots, N\}$. 
The indicator function is denoted as $\indicator{\cdot}{\cdot}$, such that $\indicator{a}{b} \! = \!1$ if $a\!=\!b$ and $0$ otherwise.
Bold letters denote vectors (e.g., $\mathbf{x}$).
We denote the space of probability measures over a finite set $E$ as $\P({E})$.
%
%
For a finite set $E$, the total variation between two probability measures $\mu, \mu' \in \P({E})$ is given by $\dtv{\mu, \mu'} = \frac{1}{2} \sum_{e\in {E}} \abs{\mu(e) - \mu'(e)} = \frac{1}{2} \norm{\mu - \mu'}_1$.


\section{Problem Formulation}
\subsection{Mean-Field Games}\label{subsec:finite-mfc}
Consider a discrete-time system with $N$ \textit{homogeneous} agents that operates over a finite horizon $T$. 
Let $X^{N}_{i,t} \in \X$ and $U^{N}_{i,t}\in \U$ denote the random variables representing the state and action taken by an agent $i \in [N]$ at time $t$. 
Here, $\X$ and $\U$ are the \textit{finite} individual state and action spaces for each agent, independent of $i$ and $t$.
The joint state and action of the population is denoted as $(\bfX^{N}_t, \bfU^{N}_t)$.
The operator $ \M_t^{N} =\empMu{\bfX^{N}_t}$ denotes the operation $\M^{N}_t(x) = \frac{1}{N} \sum_{i=1}^{N_1} \mathds{1}_x(X^{N}_{i,t})$ for all $x\in\X$ that relates joint states to the corresponding ED.
Notice that $\M^{N}_t(x)$ gives the fraction of agents at state~$x$ at time $t$.
%
%
Thus, the $\mathrm{Emp}$ operator removes agent index information and one \textit{cannot} tell the state of a specific agent~$i$ given only the ED.
We consider weakly-coupled dynamics, where the dynamics of each individual agent is coupled with other agents through the ED. 
For agent $i$, its stochastic transition is governed by the transition kernel $f_t(x_{i,t+1}^{N}\vert x_{i,t}^{N}, u_{i,t}^{N}, \mu^{N}_t)$ where $\mu^{N}_t = \empMu{\bfx_t^{N}}$.
Moreover, each agent receives at each time instant $t$ a reward determined by the reward function $r_t: \X\times \U\times \P(\X)  \to [-R_{\max}, R_{\max}]$.

\begin{assumption} 
    \label{assmpt:lipschitiz-model}
    For all $x\in \X, u\in\U$, $\mu, \mu' \in \P(\X)$, $\nu, \nu' \in \P(\Y)$ and all $t \in \{0,\ldots, T-1\}$,
    there exist constants $L_{f}, L_r> 0$ such that such that $\sum_{x' \in \X} \abs{f_t(x'|x,u,\mu) - f_t(x'|x,u,\mu')}  \leq L_{f}\dtv{\mu, \mu'}$ and $\abs{r_t(x,u,\mu) - r_t(x,u,\mu')}  \leq L_{r}\dtv{\mu, \mu'}$.
\end{assumption}

The above regularity assumptions are standard in the mean-field literature~\citep{cui2021approximately,huang2006large}, and ensure that shared identical policies scale well with population size. Without such an assumption, enforcing identical policies can lead to arbitrarily poor performance; see the counterexamples in~\citet{guan2024zero}.

For each agent $i\in [N]$ consider the following (time-varying) Markov policy $\pi_{i,t}: \U \times \X  \to [0,1]$, where $\pi_{i,t}(u|X^{N}_{i,t})$ is the probability that an agent $i$ selects action $u$ given its state $X_{i,t}^{N}$. 
Note that the policy is open-loop with respect to the ED of the agents.
An individual policy is defined as a time sequence $\pi_{i} = \{\pi_{i, t}\}_{t=0}^{T-1}$.
%
%
Let $\Pi_t$ and $\Pi$ to denote, respectively, the set of individual policies at time $t$ and across all $t$ available to each agent, i.e., $\pi_{i,t}\in \Pi_t$ and $\pi_{i} \in \Pi$ for all $i \in [N]$.
%
%
The performance of each agent $i$ is to maximize its expected cumulative reward given by the following finite-horizon objective function:

\begin{align}\label{eqn:performance-finite}
    J^{N}_i \big(\pi_1, \ldots, \pi_N\big) =
    \mathbb{E}
    \Big[\sum_{t=0}^{T-1} r_t(X^{N}_{i,t}, U^{N}_{i,t}, \M^{N}_t)\Big],
\end{align}
where $\M_t^{N} = \empMu{\bfX^{N}_t}$, and the expectation is with respect to the distribution of all system variables induced by the joint policy $(\pi_1, \ldots, \pi_N)$ and the initial state of each agents is drawn according to $x_{i, 0}\sim \mu_0$.
%

\begin{definition}
For $\epsilon \geq 0$, an $\epsilon$-Nash equilibrium is a tuple $(\pi^*_{1},\ldots, \pi^*_{N})$ such that, for all $i =1,\ldots, N$, $J^{i,N} (\pi_{i},\pi^*_{-i})\leq J^{i,N} (\pi^*_{i}, \pi^*_{-i}) +\epsilon, $ for all $\pi_{i}\in \Pi$.
When $\epsilon=0$, we have the standard Nash equilibrium.
\end{definition}

The optimization problem in~\eqref{eqn:performance-finite} is intractable for large populations, as the joint policy space $\times_{i=1}^{N} \Pi$ grows exponentially with the number of agents.
To address this scalability challenge, we consider the infinite-population limit and assume that agents deploy identical policies—a standard approximation in large-population systems that ensures tractability and robustness, albeit with some performance loss~\citep{shi2012survey, arabneydi2014team}.
We first introduce the class of identical policies.
\begin{definition} [Identical Policy]
    \label{def:identical-policy}
    The  joint policy $\pi = \{\pi_{1},\ldots, \pi_{N}\}$ is an identical policy, if $\pi_{i_1,t} = \pi_{i_2,t}$ for all $i_1,i_2 \in [N]$ and $t \in \{0, 1, \ldots, T-1\}$. 
\end{definition}

We slightly abuse notation and use $\pi$ to denote the identical joint policy, where all agents apply the same policy $\pi$.

\begin{assumption}\label{assmpt:identical-policy}
  All agents follow an identical policy $\pi$.
\end{assumption}
Consequently, we can model the behavior of the entire population as the distribution of the \textit{representative agent}, i.e., the \textit{mean-field.}
We define the mean-field (MF) as the state distribution of a typical agent in an infinite-population limit. 

\begin{definition}[Mean-Field]\label{def:mean-field}
    Given an identical policy $\pi \in \Pi$, the mean-field is defined as the sequence of vectors with initial condition $\mu_0$ that follow deterministic dynamics:
    \begin{equation}\label{eqn:mf-dynamics-local}
        \mu_{t+1} (x') \!= \!\sum_{x \in \X}  \Big[\sum_{u \in \U} f_t(x'|x, u, \mu_t)\pi_t(u|x) \Big] \mu_t(x).
    \end{equation}
\end{definition}
%

We refer to the time sequence $\mu= \{\mu_t\}_{t=0}^T \in \M \triangleq \prod_{t=0}^T \M_t$ induced by some fixed policy $\pi$ for a given initial distribution $\mu_0$ as the \textit{mean-field flow}.
At this infinite-population limit, the optimization problem 
\eqref{eqn:performance-finite} follows:
\begin{align}\label{eqn:performance-infinite}
    J_\mu \big(\pi\big) =
    \mathbb{E}
    \Big[\sum_{t=0}^{T-1} r_t(x_t, u_t, \mu_t) \Big],
\end{align}
where the expectation is with respect to the distribution of all system variables induced by the identical policy $\pi$ and the initial state follows $x_0\sim\mu_0$.
Thus, for a fixed MF flow $\mu$, we obtain an induced \textit{inhomogeneous} MDP with transitions $f_t(x, u, \mu_t)$ and rewards $r_t(x, u, \mu_t)$ with the objective~\eqref{eqn:performance-infinite}.
The standard Q-function associated with this induced MDP with fixed MF flow $\mu$ under the policy $\pi$ at time-step $t$ can be computed as:
\begin{equation}\label{standard-Q}
    Q^\pi_{\mu, t}(x, u) = r_t(x, u, \mu_t) + \sum_{x'\in\X}f_t(x'|x, u, \mu_t)V^\pi_{\mu, t+1}(x'),
\end{equation}
where $Q^\pi_{\mu, T-1}(x, u) \triangleq r_{T-1}(x, u, \mu_{T-1})$.
The value function of the induced MDP at time $t$ is given by $V^\pi_{\mu, t}(x) = \expct{\sum_{t'=t}^{T-1}r_{t'}(x_{t'}, u_{t'}, \mu_{t'}) ~|~ x_t =x} {\pi}$ or alternatively,
\begin{align}\label{standard-V}
   V^\pi_{\mu, t}(x) = \expct{Q^\pi_{\mu, t}(x, u)} {u\sim\pi_t}.
\end{align}

Consequently, the optimal policy computed in~\eqref{eqn:performance-infinite} depends on the MF flow $\mu$. 
We use the operator $\B_\opt :\M \to \Pi$ to denote the mapping from the fixed MF flow to an optimal policy of the induced MDP in~\eqref{eqn:performance-infinite}, i.e., $ \pi^{*} \in \B_\opt \left(\mu\right) \triangleq\arg\max_{\pi\in\Pi}J_\mu \big(\pi\big)$.
When all agents employ the policy $\pi$ of the representative agent, a new MF flow $\mu$ is induced and can be propagated via~\eqref{eqn:mf-dynamics-local} starting from $\mu_0$. 
We use the operator $\B_\prop : \Pi \to \M$ to denote this propagation, that is, $ \mu = \B_\prop \left(\pi \right)$.
Using these operators, we define the mean-field equilibrium (MFE).
\begin{definition}[MFE~\citep{huang2006large}]\label{def:mfe}
    The MFE is defined as a consistent pair $(\pi^*, \mu^*) \in \Pi \times \M$ such that $ \pi^* \in \B_\opt \big( \mu^* \big)$ and $\mu^* = \B_\prop(\pi^*)$.
\end{definition}
Under Assumption~\ref{assmpt:lipschitiz-model}, the existence of such consistent pair can be established through Kakutani's fixed point argument~\citep{cui2021approximately}. 
Furthermore, the optimal identical policy $\pi^*$ applied by all agents in the $N$-player game preserves $\epsilon$-optimality in the finite-population system~\citep{guan2022shaping}.



\subsection{Risk due to Uncertainty in Initial Distributions}\label{subsec:risk-aversion}

%

%
Once determined, the optimal policy $\pi^*$ 
is open-loop and operates independently of the realized mean-field information. 
Consequently, a policy optimized for a specific initial distribution may perform poorly when deployed under a different initial condition.
{To improve robustness, alternatively, one may na\"{\i}vely recompute policies for each initial condition.
However, this would require real-time access to MF information to know which initial condition has been realized, which is often costly or infeasible due to communication, sensing, or privacy constraints. 
Thus, we seek robustness in the class of policies that are open-loop with respect to the MF.}
These considerations highlight the need to explicitly account for uncertainty in the initial distribution as illustrated by the example below.

\textbf{Motivating Example 1:}
Consider a population-level report estimating the proportion of individuals infected by a contagious disease. 
Based on this estimate, individuals choose whether to quarantine or continue normal activities, incurring costs from treatment if infected or from lost income and social isolation if quarantined.
If the report incorrectly indicates zero initial infection due to biased or incomplete data, individuals rationally avoid quarantine, even though the true infection rate is non-zero.
Recomputing the consensus estimate may be prohibitively expensive in terms of operational effort and human resources. 
Consequently, individuals who choose not to quarantine are exposed to an unanticipated risk of infection.

This example highlights the critical importance of accounting for uncertainty in the initial population distribution.
A risk-averse decision-making approach, in which individuals hedge against potential inaccuracies in the reported initial distribution, would have resulted in more robust behavior and mitigated exposure to adverse and unforeseen outcomes.

Formally, consider a finite set $\mathbb{M}$ of initial 
distributions
\footnote{%
Restricting to a finite set of initial distributions simplifies both analysis and computation by reducing expectations to finite sums rather than integrals.
Extending the framework to continuous and compact sets $\mathbb{M}$ is left for future work.
}.
Let the underlying probability distribution over the set of initial mean-field distributions $\mathbb{M}$ be denoted by $\Gamma^*_\mathbb{M} \in \P({\mathbb{M}})$. 
%
%
For a given initial distribution $\mu_0 \in \mathbb{M}$, let the pair $(\mu^*, \pi^*)$ constitute an MFE, where $\mu^*$ denotes the optimal mean-field flow induced by the policy $\pi^*$ under the initial condition $\mu_0$.
%
%
%
%
If the system is initialized with a different $\bar \mu_0 \in \mathbb{M}$, there is no guarantee that the pair $(\bar{\mu}, \pi^*)$ constitutes an MFE, where $\bar{\mu}=\B_\prop(\pi^*)$.
Under the ``unexpected" or erroneous initial condition $\bar{\mu}_0$, following the old $\pi^*$ often leads to suboptimal performance.
%
%

%
A natural response is to optimize expected performance by averaging over initial conditions induced by $\Gamma^*_\mathbb{M}$ but such objectives fail to capture population-level risk preferences.
For example, in the epidemiological setting of Example~1, consider the disease-free distribution $\mu_0^a$ occurring with probability $0.99$ while $\mu_0^b$ representing a small but persistent infected population, occurs with probability $0.01$. 
Optimizing expected performance effectively ignores $\mu_0^b$, despite its potential for severe public-health consequences.
This highlights the inadequacy of expectation-based objectives in the presence of rare but high-impact events and motivates the incorporation of risk aversion.

\textbf{Risk Averse Design.} Introducing risk aversion entails a trade-off between maximizing average performance and achieving reliable outcomes under uncertainty.
In settings with sensor noise, incomplete information, or adversarial perturbations of the initial distribution, such robustness considerations are essential, and agents may prefer conservative policies over purely reward-maximizing ones. 
Accordingly, we seek policies that explicitly minimize risk with respect to uncertainty in the initial MF distribution.

There are several risk metrics that arise in game theory, especially from operations research and finance. 
Following~\cite{mazumdar2025tractable}, we focus on convex measures of risk whose formal definition is given below:
\begin{definition}[Convex Risk Measures]\label{def:convexmetric}
    Let $\mathcal{F}$ be the set of functions mapping from a space of outcomes $\Omega$ to $\mathbb{R}$. A convex measure of risk is a mapping $\rho:\mathcal{F}\rightarrow \mathbb{R}$ satisfying: $(i)$ \emph{Monotonicity:} If $F\le G$ almost surely, then $\rho(F)\ge \rho(G)$, $(ii)$ \emph{Translation Invariance:} If $m\in \mathbb{R}$ then $\rho(F+m)=\rho(F)-m$, and $(iii)$ \emph{Convexity:} For all $\lambda \in (0,1)$, $\rho(\lambda F+(1-\lambda)G)\le \lambda\rho(F)+(1-\lambda)\rho(G)$.
   
\end{definition}

As a result, we say that $F(\omega)$ is ``riskier" than $G(\omega)$ under $\rho$ when the associated evaluation of $F(\omega) \leq G(\omega)$.
Convex risk measures have well-posed, intuitive properties as seen in Definition~\ref{def:convexmetric}.
Furthermore, convex risk measures
admit a dual representation as described in the following theorem.

\begin{theorem}[Dual Representation Theorem~\citep{Risk_overview}]\label{thm:dual-representation}
Let $\mathcal{X}$ be the set of functions mapping from a finite set $\Omega$ to $\mathbb{R}$. Then a mapping $\rho:\mathcal{X}\rightarrow \mathbb{R}$ is a convex risk measure if and only if there exists a penalty function $D:\P(\Omega)\rightarrow (-\infty,\infty]$ such that:
$\rho(X)=\sup_{p\in \P(\Omega)} \mathbb{E}_p[-X]-D(p),$
where $\P(\Omega)$ is the set of all probability measures on $\Omega$. Furthermore, the function $D(p)$ can be taken to be convex, lower-semi-continuous, and satisfy $D(p)>-\rho(0)$ for all $p\in \P({\Omega})$. 

\end{theorem}


%
%

%
%
Define the fixed set of mean-field flows $\S_\mathbb{M} = \{\mu^k\}_{k=1}^{|\mathbb{M}|}$ where $\mu^k$ corresponds to the mean-field flow with the initial distribution $\mu^k_0\in \mathbb{M}$.
Uncertainties arising due to the stochasticity in the selection of the initial condition from the fixed set $\mathbb{M}$ result in the following risk-averse objective function at a given time step $t$ for a given state $x$:
\begin{align}\label{mf-flow-risk-averse-V}
    c_t^\pi(x; \S_\mathbb{M}) = \rho_{\Gamma^*_\mathbb{M}}\big(V^\pi_{\mu, t}(x) \big),
\end{align}
where $V^\pi_{\mu, t}(x)$ is computed based on~\eqref{standard-V} for each $\mu_0 \in \mathbb{M}$ and $\rho_{\Gamma^*_{\mathbb{M}}}$ is the convex risk measure for $\Gamma^*_{\mathbb{M}}$, i.e., the uncertainty in the initial distribution.
Furthermore, the mean-field flow $\mu$ is treated as an exogenous signal and is propagated under the given policy $\pi$ using the $\B_{\prop}$ operator. 
%


%
\begin{remark}
{As emphasized earlier, we focus on decentralized policies that depend only on local state information and do not require real-time MF feedback.
Consequently, the risk-averse objective must be formulated and solved at each time step to ensure risk aversion over the entire horizon.
This is necessary because the realized initial MF distribution is unknown (due to the absence of MF feedback), requiring robustness to consistently be maintained over time.}
\end{remark}
Following Theorem~\ref{thm:dual-representation}, we can rewrite~\eqref{mf-flow-risk-averse-V} for the fixed set of mean-field flow $\S_\mathbb{M}$, with $\Omega = \mathbb{M}$ as
\begin{align}\label{eq:dual-rep}
     &c_t^\pi(x; \S_\mathbb{M}) \nonumber =  \rho_{\Gamma^*_\mathbb{M}}(V^\pi_{\mu, t}(x))\nonumber\\
     &\equiv \sup_{\hat{\Gamma}_\mathbb{M}\in\P(\mathbb{M})} -\expct{\pi_t\t(\cdot | x)Q^\pi_{\mu, t}(x, \cdot)}{\mu_0\sim\hat{\Gamma}_\mathbb{M}} -  \frac{1}{\tau}D(\hat{\Gamma}_\mathbb{M}, \Gamma^*_\mathbb{M}).
\end{align}
%
Agents now seek to compute a policy that minimizes the risk in~\eqref{eq:dual-rep}.
%
%
In this form, one can see that in a risk-averse setting, the population imagines that fictitious adversaries seek to
maximize their cost $c_t^\pi(x)$ at each time-step $t$ and state $x$ but are penalized from deviating too far from the realized initial distribution $\Gamma^*_\mathbb{M}$.
As $\tau\to\infty$ the population becomes increasingly risk-averse—and, in the extreme, the initial distribution selection  can be treated as entirely adversarial.
%
For simplicity, we use KL-divergence as the penalty function, 
since this choice admits the following closed form solution for the optimization problem in~\eqref{eq:dual-rep}:
\begin{align*}
    c_t^\pi(x; \S_\mathbb{M})  &=  \frac{1}{\tau}\log \Big(\sum_{k=1}^{|\mathbb{M}|}\Gamma^*_\mathbb{M}(\mu_0^k)\exp^{-\tau\pi_t\t(\cdot | x)Q^\pi_{\mu^k, t}(x, \cdot)}\Big).
\end{align*}
Note that this is an explicit function of $\pi_t(\cdot | x)$ for all states $x\in \X$ and time-steps $t\in\{0, \ldots, T-1\}$.
%

\subsection{Bounded Rationality in Mean-Field Games}

{Definition~\ref{def:mfe} assumes perfect rationality}, requiring agents to compute optimal policies with complete knowledge of the environment and other agents’ behavior.
In practice, such assumptions are often violated due to information, computational, or time constraints, motivating the use of \emph{bounded rationality} to model approximate decision-making.
Indeed, \citet{goeree2003risk} argue that human decision-making and strategic behavior are more accurately captured by a combination of risk aversion and bounded rationality than by either property in isolation. 

Broadly, two predominant approaches have emerged in the literature in terms of quanitfying bounded rationality.
One approach to bounded rationality introduces stochastic perturbations into agents’ decision processes, yielding quantal response policies that can equivalently be obtained by augmenting the objective with a strictly convex regularizer~\citep{mazumdar2025tractable, goeree2016quantal}.
A second approach is based on cognitive hierarchy models~\citep{shao2019multimedia}, which construct a hierarchy of policies where each level best responds to lower-level strategies.
We adopt the former approach, as it integrates naturally with our risk-averse formulation.
Formally, the combined objective in our case of initial distribution-dependent risk takes the form:
\begin{align}\label{eq: BR-RA-V}
    c_t^{\pi, \alpha}(x; \S_\mathbb{M}) = \rho_{\Gamma^*_\mathbb{M}}(V^\pi_{\mu, t}(x)) + \alpha\nu(\pi_t(\cdot | x)),
\end{align}
for all $x\in\X$ and $t\in\{0, \ldots, T-1\}$, where $\nu$ is a suitable (strict) convex regularizer with temperature parameter $\alpha>0$ capturing the degree of bounded rationality.

In summary, a \emph{Risk-Averse Quantal Response Mean-Field Game} (RQ-MFG) is specified by the game tuple: $\langle \X, \U, T, f, r,  \nu,\mathbb{M}, \Gamma^*_\mathbb{M},  D \rangle$.
%
The additional elements $(\mathbb{M}, \Gamma^*_\mathbb{M}, D)$ capture the risk aversion with respect to the initial MF distribution,
while the convex regularizer $\nu$  is associated with bounded rationality. 
%
Thus, by integrating both risk-aversion and bounded rationality, our formulation yields a more realistic modeling framework that helps bridge the gap between theoretical abstractions and observed behavior.

\section{MF-RQE}

Our focus shifts from reward maximization in~\eqref{eqn:performance-infinite} to minimizing the risk defined in~\eqref{eq: BR-RA-V}.
Define $\pi_{-t} = \{\pi_0, \!\ldots\!, \pi_{t-1}, \pi_{t+1}, \!\ldots\!, \pi_{T-1}\}$ as the set of policies for all times other than $t$.
For a fixed set of mean field flows $\S_\mathbb{M}$ originating from the initial distributions $\mu_0\in \mathbb{M}$, we seek to obtain a policy $\pi^{*}\in\Pi$ such that, for all $\pi\in\Pi$, $t\in\{0, 1, \ldots, T-1\}$ and $x\in\X$.
%
\begin{align}\label{eqn:BR-RA-opt}
    c_t^{\pi^{*}, \alpha}(x; \S_\mathbb{M}) &\leq c_t^{(\pi'_t, \pi^*_{-t}), ~\alpha}(x; \S_\mathbb{M})\nonumber\\
    \Rightarrow\pi^*_t(\cdot | x) &= \arg\min_{\pi_t\in\Pi_t} c_t^{(\pi_t, \pi^*_{-t}), ~\alpha}(x; \S_\mathbb{M}).
\end{align}
%

Define the {set operators $\B^{\textrm{RQE}}_\opt :2^\M \to \Pi$ to denote the mapping from the fixed set of mean-field flows $\S_\mathbb{M}\subseteq \M$ or equivalently, $ \S_\mathbb{M}\in 2^\M $, to optimal policies of~\eqref{eqn:BR-RA-opt}}.
In other words, $ \pi^{*} = \B^{\textrm{RQE}}_\opt \left(\S_\mathbb{M}\right)$.
%
When all agents employ a policy $ \pi$ of the representative agent, a new set of mean field flows are induced and can be propagated via~\eqref{eqn:mf-dynamics-local} starting from $\mu_0\in \mathbb{M}$. 
We use the operator $\B^{\textrm{RQE}}_\prop : \Pi \to 2^\M $ to denote this propagation. That is, $\S_\mathbb{M} \in \B^{\textrm{RQE}}_\prop \left(\pi \right)$, corresponding to all $\mu_0\in \mathbb{M}$.
%
Using these operators, we define the notion of the mean-field risk-averse quantal response equilibrium (MF-RQE).

\begin{definition}[MF-RQE]\label{definition:mf-rqe}
    For a fixed set of initial mean-field distributions $\mathbb{M}$, the MF-RQE is defined as the consistent pair $(\pi^{*}_{\textrm{RQE}}, \S^*_{\mathbb{M}})$ such that  {$\pi^*_{\textrm{RQE}} =  \B^{\textrm{RQE}}_\opt\big( \S^*_{\mathbb{M}} \big)$ and $\B^{\textrm{RQE}}_\prop(\pi^*_{\textrm{RQE}})=\S^*_{\mathbb{M}}$},
where  {$\S^*_{\mathbb{M}}$} is the set of mean-field flows corresponding to all the initial conditions $\mu_0\in \mathbb{M}$.
\end{definition}

\begin{proposition}\label{prop:existence-MF-RQE}
    Under Assumptions~\ref{assmpt:lipschitiz-model}-\ref{assmpt:identical-policy}, there exists an MF-RQE $(\pi^*_{\textrm{RQE}}, \S^*_{\mathbb{M}})$.
\end{proposition} 

Having motivated and developed the risk-averse quantal response mean-field game (RQ-MFG), we now establish its connection to the corresponding finite-population game comprising $N$ agents. 
In this finite-population setting, risk aversion is defined with respect to uncertainty in the initial state distribution, consistent with the MF formulation, i.e., the joint state $\bfX_{t=0}$ is sampled from $\mu_0$, where $\mu_0\sim\mathbb{M}$.
For each agent $i \in [N]$, the interaction is modeled as a risk-averse Markov game at each time step $t \in {0,\ldots,T-1}$ and state $x \in \mathcal{X}$.
Given a joint policy $\pi = (\pi_1, \ldots, \pi_N)$, the stage cost for agent $i$ is defined as
\begin{align*}
c_t^{\alpha, i, N}(\pi_1,\!\ldots\!,\pi_N; x)
\!=\!\rho_{\Gamma^*_\mathbb{M}}\left( V^\pi_{i,t}(x) \right)
\!+\! \alpha \nu\big(\pi_{i,t}(\cdot|x)\big),
\end{align*}

where $V^\pi_{i,t}(x)$ denotes the value function associated with the finite-population objective~\eqref{eqn:performance-finite} for agent $i$ under the joint policy $\pi$, and where $\rho_{\Gamma^*_\mathbb{M}}$ and $\nu$
%
%
are consistent with~\eqref{eq: BR-RA-V}, together defining a risk-averse and boundedly rational objective at each time step.
The solution concept is defined below.
%


\begin{definition}

For $\epsilon \geq 0$, an $\epsilon$-RQE is a tuple $(\pi^*_1,\ldots, \pi^*_{N})$ such that, for all $x\in\X$, all $t \in \{0,\ldots, T-1\}$, and all agents $i\in[N]$ and for all $\pi_{i, t}\in \Pi_t$, {it follows that 
$c_t^{\alpha, i, N}\big(\pi^*_{i}, \pi^*_{-i}; ~x\big) -\epsilon\leq c_t^{\alpha, i, N}\big((\pi_{i, t}, \pi^*_{i, -t}), \pi^*_{-i}; ~x\big)$.}
For $\epsilon=0$, we have the standard risk-averse quantal-response equilibrium (RQE)~\citep{mazumdar2025tractable}.
\end{definition}

The following theorem shows that solving the limiting RQ-MFG yields an approximate solution to the finite-population RQ-game when the number of agents $N$ is large, thereby  establishing a connection between the infinite-population RQ-MFG and the corresponding finite-population game.
Consequently, this result overcomes the \emph{curse of dimensionality} limitation identified in~\cite{mazumdar2025tractable}.

\begin{theorem}\label{thm:MF-RQE-finite-population}
    Under Assumptions~\ref{assmpt:lipschitiz-model}-\ref{assmpt:identical-policy}, if $(\pi^*_{\textrm{RQE}}, \S^*_{\mathbb{M}})$ is an MF-RQE, then the policy $(\pi^*_{\textrm{RQE}}, \pi^*_{\textrm{RQE}}, \ldots, \pi^*_{\textrm{RQE}})$ is an $\epsilon$-RQE in the finite $N$-agent game.
\end{theorem}
\section{Learning MF-RQE}
Several algorithms have been proposed to compute Nash equilibria in MFGs, including fixed-point iteration (FPI), fictitious play (FP), and online mirror descent (OMD)~\citep{cui2021approximately, perolat2021scalingmeanfieldgames, lauriere2022scalable}. 
These methods typically rely on restrictive assumptions such as monotonicity, mean-field–independent dynamics, risk neutrality, or specific entropy-based regularization, which apply to a restricted class of mean-field dynamic models and do not consider population-level risk preference.
In contrast, our framework combines bounded rationality—implicitly inducing convex regularization—with explicit risk modeling via convex risk measures, enabling richer dynamics and robustness to uncertainty in the initial mean-field distribution. 
The resulting solution concept, the MF-RQE, arises naturally from this formulation. 
%
In this section, we extend FPI and FP to compute MF-RQE under known dynamics and rewards. 
For large-scale or partially unknown environments, we propose a scalable actor–critic-based RL approach that learns MF-RQE policies from data samples.

\textbf{Fixed Point Iteration.} We compute the MF-RQE for a prescribed degree of bounded rationality $\alpha$ and risk aversion $\tau$ as defined in~\eqref{eq: BR-RA-V}. 
To this end, we employ a fixed-point iteration (FPI) scheme that computes the MF-RQE policy $\pi^{\star}_{\textrm{RQE}}$ as a fixed point of the composite operator $\B^{\textrm{RQE}}_\opt\circ \B^{\textrm{RQE}}_\prop$, 
given a set of initial MF distributions $\mathbb{M}$ and their associated probabilities $\Gamma^*_{\mathbb{M}}$.
The algorithm is initialized with an arbitrary policy, typically uniform action selection, i.e., $\pi_t(u|x) \!=\! 1 / |\mathcal{U}|$.
At each iteration, the operator $\mathcal{B}^{\textrm{RQE}}_\opt \circ \mathcal{B}^{\textrm{RQE}}_\prop$ is applied to update the policy, and this process is repeated until convergence to a fixed point is achieved.
For a fixed MF flow, the operator $\mathcal{B}^{\textrm{RQE}}_{\mathrm{opt}}(\cdot)$ can be computed efficiently via dynamic programming. 
%


We impose the convexity assumption on the regularization function $\nu$ to facilitate the convergence analysis.

\begin{assumption}\label{assmpt: strongly-convex-nu}
$\nu(\pi)$ is $m$-strongly convex.
\end{assumption}

The standard entropy regularization satisfies Assumption~\ref{assmpt: strongly-convex-nu} with $m = 1$.
Although strong convexity is required for our theoretical guarantees, empirical evidence suggests that the proposed algorithms continue to perform well under weaker conditions, such as strict convexity of $\nu$.

\begin{theorem}\label{thm:FPI-convergence}
    Under Assumptions~\ref{assmpt:lipschitiz-model}-\ref{assmpt: strongly-convex-nu}, FPI converges to an MF-RQE for sufficiently large $\alpha$.
\end{theorem}

\textbf{Fictitious Play.} An alternative method for computing Nash equilibria is the fictitious play algorithm.
It utilizes averaging of policies or distributions to achieve convergence~\citep{cardaliaguet2015learningmeanfieldgames}.
\citet{perrin2020fictitious} establish convergence of fictitious play under monotonicity and MF–independent transition dynamics. 
\citet{eich2025bounded} later proposed Generalized Fictitious Play, a discrete-time extension and empirically demonstrate convergence to a regularized equilibrium in the presence of MF-dependent transition dynamics.
Motivated by these results, we extend fictitious play to the RQ-MFG setting and introduce \emph{Risk-Averse Quantal Fictitious Play} (RQ-Fictitious Play), which computes an MF-RQE via policy averaging. 
%
%

\begin{theorem}\label{thm:Fictitious-Play-convergence}
    Under Assumptions~\ref{assmpt:lipschitiz-model}-\ref{assmpt: strongly-convex-nu}, RQ-Fictitious Play converges to an MF-RQE for sufficiently large $\alpha$.
\end{theorem}

\textbf{Deep Reinforcement Learning.} 
{
We incorporate deep reinforcement learning to extend fixed-point iteration to model-free settings.
Both RQ-FPI and RQ-Fictitious Play require full knowledge of the RQ-MFG game tuple and rely on tabulated value functions, making them impractical when the game model is unknown or too complex to reconstruct from empirical data.
We therefore propose Deep Risk-Averse Quantal Fixed-Point Iteration (D-RQ-FPI), a model-free, function-approximation–based alternative.
}


%
The proposed algorithm maintains $|\mathbb{M}|$ critics $Q(\cdot;\theta_k)$, each corresponding to an initial distribution $\mu_0^k\in\mathbb{M}$, and a time-dependent policy $\pi(t,x; \zeta) := \pi_\zeta(t,x)$. 
Mean-field flows are computed analytically or via stochastic simulation~\citep{cui2021approximately}, and the goal is to learn a best-response policy to the induced flows.
Trajectories are collected by sampling $\mu_0^k \sim \Gamma^*_{\mathbb{M}}$ and rolling out the policy under the corresponding induced mean-field flow.
Each critic is updated using a temporal-difference loss, $\mathcal{L}(\theta_k) = \mathbb{E}_{x,u,x',r}[(Q^{\pi_\zeta}(x, u; \theta_k) - y_k)^2]$ where $y_k = r_k + (1 - d), Q^{\pi_\zeta}(x', u'; \theta'_k)$, $\theta_k'$ denotes target network parameters and $d$ denotes the termination signal.
%
%
Given the updated critics, the actor is optimized to solve the risk-averse quantal response objective in~\eqref{eq: BR-RA-V} for each state–time pair $(x,t)$, thereby computing an approximate best response to the current set of MF flows. 
The updated policy induces a new collection of MF flows, and the procedure is repeated until convergence.



\section{Numerical Experiments}

In this section, we evaluate the proposed algorithms on a collection of Risk-Averse Quantal Mean-Field Games (RQ-MFGs). 
Experiments are conducted on benchmark environments from \texttt{MFGLib}~\citep{guo2023mfglib}, which includes the epidemiological SIS game introduced in Section~\ref{subsec:risk-aversion}, as well as a newly proposed one-dimensional congestion game.

Following~\cite{eich2025bounded}, convergence to an MF-RQE is assessed using an exploitability-like metric,
\begin{align}\label{eq:BR-RA-explt}
    \Delta c(\pi)
    \triangleq
    \max_{t \in \mathcal{T}}
    \left\|
    \pi_t
    -
    \B^{\textrm{RQE}}_{\opt}
    \circ
    \B^{\textrm{RQE}}_{\prop}(\pi)_t
    \right\|,
\end{align}
which serves as our primary measure of efficacy of the proposed algorithms.
Smaller values of $\Delta c(\pi)$ indicate reduced sensitivity to adverse realizations of the initial mean-field distribution and the associated underlying probability distribution $\Gamma^*_\mathbb{M}$, and thus a lower degree of exploitability under distributional uncertainty.
Conceptually, this metric assesses whether, for a given policy and its induced set of MF flows, there exists an alternative policy that is \emph{less risky} and \emph{more robust} with respect to uncertainty in those flows, where risk is quantified through the objective in~\eqref{eqn:BR-RA-opt}.

To highlight the benefits of the MF-RQE framework, we compare MF-RQE policies against: (i) entropy-regularized Nash equilibrium (NE) policies~\citep{cui2021approximately} computed under a single initial distribution and evaluated under the risk-averse objective over $\mathbb{M}$, and (ii) a risk-neutral policy {$\pi^{*}_{\textrm{avg}}$ that maximizes expected cumulative reward averaged over initial distributions.}
%
These baselines are evaluated under uncertainty in the realized initial distribution, thereby highlighting the limitations of {risk-neutral policies} when deployed in settings with distributional uncertainty.

Our results show that incorporating risk aversion with respect to the initial mean-field distribution—together with bounded rationality via general convex regularizers—yields policies that are robust to distributional uncertainty and non-exploitable in the sense of~\eqref{eq:BR-RA-explt}. 
In contrast, risk-neutral policies exhibit higher exploitability when the realized initial distribution deviates from the nominal one. 
Policy performance is further evaluated via average returns over 10,000 episodes, aggregated across five random seeds under the optimal mean-field flows $\mathcal{S}^*_{\mathbb{M}}$, illustrating the robustness–return trade-off.
Furthermore, RQ-Fictitious Play recovers the same policies as RQ-FPI across all tested environments.

In the deep RL setting, finite-sample stochasticity prevents exploitability from converging exactly to zero. 
Accordingly, following~\cite{lauriere2022scalable}, we additionally measure convergence via the distance $d_\S(\S^{\textrm{RQ-FPI}}_\mathbb{M}, \S^{\textrm{D-RQ-FPI}}_\mathbb{M})$%
\footnote{
$d_\S(\S^1_\mathbb{M}, \S^2_\mathbb{M}) \!\triangleq\! \max_{k,t}\dtv{\mu^{1,k}_t, \mu^{2,k}_t}$, where $\mu^{i,k}\!\in\! \S^i_\mathbb{M}$.
}
between empirical and analytical mean-field flows obtained via RQ-FPI,
which provides a complementary notion of convergence under sampling noise and function approximation.

\textbf{SIS Game.} The SIS game environment in \texttt{MFGLib} builds on the discussion in Section~\ref{subsec:risk-aversion}. 
%
%
%
In order to reformulate this into an RQ-MFG, we set $\mathbb{M} \!=\! \{\mu_0^1, \mu_0^2, \mu_0^3, \mu_0^4\}$ where $\mu_0^1 \!=\! [1, 0], \mu_0^2 \!=\! [0.5, 0.5], \mu_0^3 \!=\! [0, 1] , \mu_0^4 \!=\! [0.8, 0.2]$ and the probability of these initial distributions as $\Gamma^*_\mathbb{M}(\mu_0^1) = \Gamma^*_\mathbb{M}(\mu_0^4) = 0.3$ and $\Gamma^*_\mathbb{M}(\mu_0^2) = \Gamma^*_\mathbb{M}(\mu_0^3) = 0.2$.

Table~\ref{table:fpi-results} reports the exploitability metrics in~\eqref{eq:BR-RA-explt} and expected returns of the MF-RQE policy computed using the RQ-FPI algorithm with KL-penalty and entropy-regularization, policies optimized with respect to the expected initial distribution  and $\pi^{*}_{\textrm{avg}}$. 
%
The optimal policy computed using the RQ-Fictitious Play coincides with the solutions obtained from RQ-FPI.
The results demonstrate that the algorithm converges to an MF-RQE policy that is non-exploitable; together with the corresponding mean-field flows, it constitutes an MF-RQE pair $(\pi^{*}_{\textrm{RQE}}, \mathcal{S}^{*}_{\mathbb{M}})$.
Although the policy $\pi^{*}_{\mu_0^1}$ achieves higher expected returns, it is exploitable under uncertainty in the initial mean-field distribution.
In particular, there exists an alternative policy that achieves a lower risk-averse cost~\eqref{eq: BR-RA-V} when evaluated against $\pi^{*}_{\mu_0^1}$ over the set of initial distributions $\mathbb{M}$ and the induced mean-field flows $\B^{\textrm{RQE}}_\prop(\pi^{*}_{\mu_0^1})$. 
%
%
This also highlights the fundamental trade-off between risk-averse policy design and pure expected cumulative reward maximization.
Notably, despite this trade-off, the empirical differences in expected returns across policies remain relatively small.

The results from D-RQ-FPI are presented in Figure~\ref{fig:d-rq-fpi}.
The empirical exploitability decreases over training and falls below $10^{-2}$, while $d_\S(\S^{\textrm{RQ-FPI}}_\mathbb{M}, \S^{\textrm{D-RQ-FPI}}_\mathbb{M}) \approx 10^{-3}$.
These values are consistent in magnitude with those reported in standard MFGs~\citep{laurière2022scalabledeepreinforcementlearning}.

\textbf {Congestion Game.} %
The game consists of 4 states aligned in a 1D grid  and the actions $\U = \{\texttt{LEFT, RIGHT, STAY}\}$ for \textit{deterministically} moving in the respective directions.
The running reward is given by $r_t(x, u) = -2\mu(x) -0.1\mathds{1}_{u\neq \texttt{STAY}}$, discouraging agents concentrating at a single state.
%
%
The four possible initial distributions are $ \mu_0^1 \!=\! [0.25, 0.25, 0.25, 0.25], \mu_0^2 \!=\! [1, 0, 0, 0], \mu_0^3 \!=\! [0.1, 0.5, 0.2, 0.2], \mu_0^4 \!=\! [0, 0.6, 0.4, 0]$, each with probability, 0.4, 0.1, 0.3, and 0.2, respectively. 

%

\begin{table}[ht]

\centering
\caption{Computing $\Delta c(\pi)$ and Empirical Returns}
\footnotesize{
\setlength{\tabcolsep}{3.5pt}
\begin{tabular}{c cc cc}
\toprule
& \multicolumn{2}{c}{\textbf{SIS}} & \multicolumn{2}{c}{\textbf{Congestion}} \\
\cmidrule(lr){2-3} \cmidrule(lr){4-5}
Policy 
& $\Delta c(\pi)$ & Expected Returns
& $\Delta c(\pi)$ & Expected Returns \\
\midrule
$\pi^{*}_{\mu^1_0}$ & 0.090 & $\mathbf{-22.003 \pm 0.061}$ & 0.022 & $-3.295 \pm 0.008$ \\
$\pi^{*}_{\mu^2_0}$ & 0.013 & $-22.189 \pm 0.057$ & 0.145 & $\mathbf{-3.290 \pm 0.008}$ \\
$\pi^{*}_{\mu^3_0}$ & 0.056 & $-22.201 \pm 0.059$ & 0.017 & $-3.295 \pm 0.008$ \\
$\pi^{*}_{\mu^4_0}$ & 0.049 & $-22.171 \pm 0.060$ & 0.036 & $-3.295 \pm 0.008$ \\
$\pi^{*}_{\textrm{avg}}$ & 0.032 & $-22.135 \pm 0.058$ & 0.004 & $-3.293 \pm 0.008$ \\
$\pi^{*}_{\textrm{RQE}}$ & \textbf{0.00} & $-22.165 \pm 0.059$ & \textbf{0.00} & $-3.294 \pm 0.008$ \\
\bottomrule
\end{tabular}
}
\label{table:fpi-results}

\end{table}

The results of the RQ-FPI algorithm are tabulated in Table~\ref{table:fpi-results}.
Consistent with the previous results, the MF-RQE policy computed using the RQ-Fictitious Play coincide with the solutions from RQ-FPI and remain non-exploitable under uncertainty over initial MF distributions.
%
This robustness, however, comes at the cost of a modest reduction in expected returns compared to the other policies. 
The results of D-RQ-FPI shown in Figure~\ref{fig:d-rq-fpi} indicate that the empirical exploitability decreases to the order of $10^{-4}$ while $d_\S(\S^{\textrm{RQ-FPI}}_\mathbb{M}, \S^{\textrm{D-RQ-FPI}}_\mathbb{M}) < 10^{-3}$.
This demonstrates that D-RQ-FPI achieves performance comparable to the analytical solution, despite the presence of sampling noises.

\textbf{Other Environments.} 
We further compute MF-RQE policies for several additional environments provided in \texttt{MFGLib} (Figure~\ref{fig:convergence}).
The empirical trends observed in the SIS and congestion games persist across these environments, underscoring the importance of explicitly accounting for risk aversion in mean-field games.
%
%
%
We extend our analysis beyond entropy regularization by considering the log-barrier regularizer $\nu(\pi(\cdot | x)) \!=\! -\sum_{u\in\U}\log(\pi(u | x))$.
%
%
%
%
Figure~\ref{fig:convergence}(c) illustrates the corresponding convergence behavior of RQ-FPI.
%
Overall, these results are consistent with those observed under entropy regularization, further demonstrating that the MF-RQE formulation naturally extends beyond entropy-based regularizers.
In contrast, ablation studies conducted without convex regularization (Figure~\ref{fig:convergence}(d)) highlight the critical role of bounded rationality: without it, neither tractability nor convergence to an MF-RQE can be guaranteed.
{
Details of problem setups and hyperparameters are presented in Appendix~\ref{appdx:experiments}.}

\begin{figure}[t!]
    \centering
    \includegraphics[width=\textwidth]{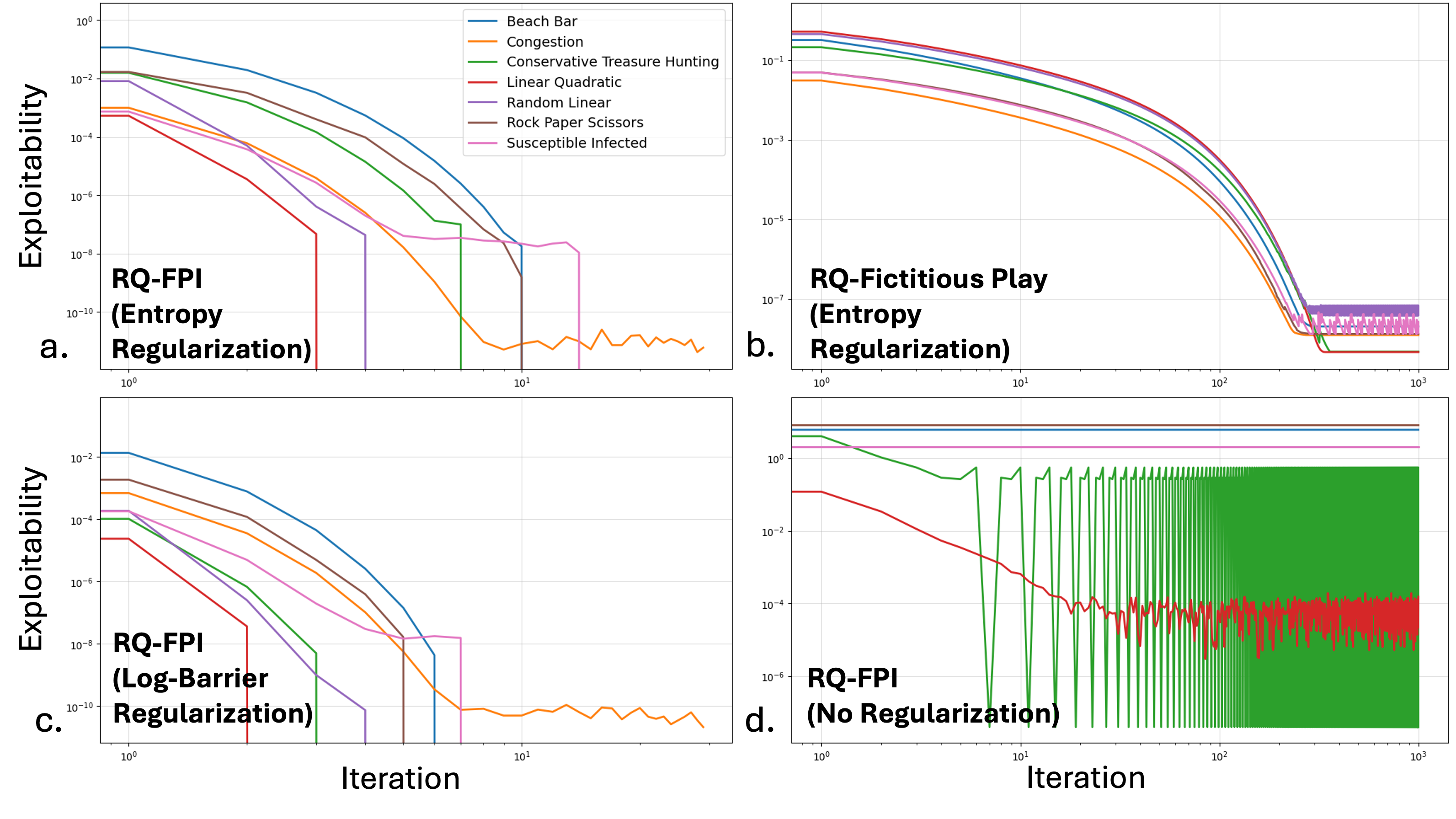}
   
    \caption{Performance of RQ-FPI and RQ-Fictitious Play across different environments and regularizers.}
    \label{fig:convergence}
    
\end{figure}

\begin{figure}[t!]
    \centering
    \includegraphics[width=\textwidth]{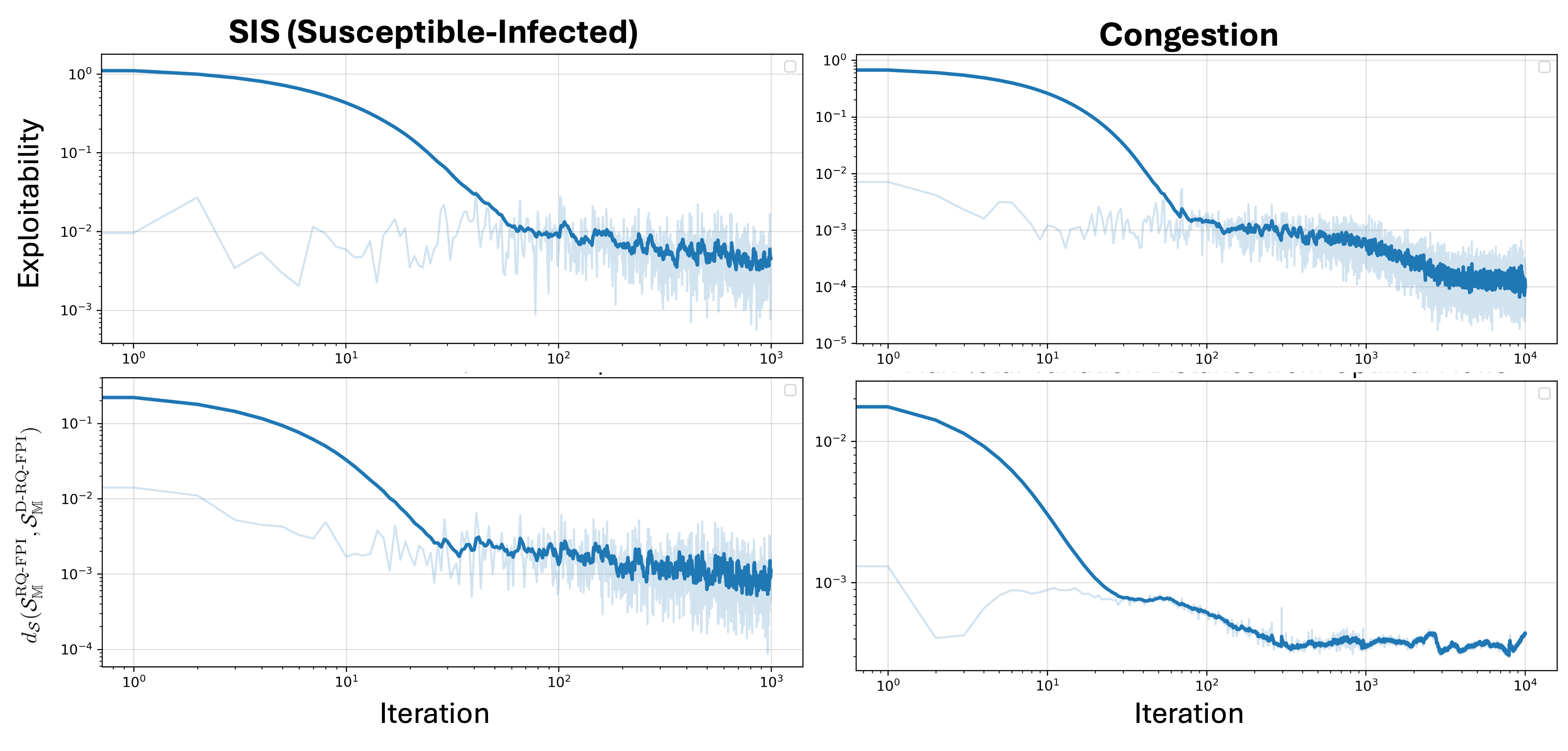}
    
    \caption{Convergence results of D-RQ-FPI with entropy regularization in the SIS and the Congestion scenarios.}
    \label{fig:d-rq-fpi}
    
\end{figure}

\section{Conclusions and Future Work}

In this work, we introduced a class of mean-field games that incorporates risk aversion with respect to a set of initial mean-field distributions while relaxing the assumption of perfect rationality. 
This leads to a tractable solution concept—the MF-RQE—that more accurately reflects large-scale, decentralized decision-making.
We established existence results, related the infinite-population equilibrium to its finite-population counterpart, and developed efficient algorithms for computing MF-RQE under both model-based and sample-based settings. 
Numerical experiments show that MF-RQE policies are robust to initial distribution uncertainty, remain non-exploitable, and achieve comparable performance despite the trade-offs induced by risk aversion.
An important direction for future work is extending this framework to heterogeneous and team-based mean-field settings, where risk aversion may arise with respect to adversarial populations, further strengthening the link between mean-field theory and observed decision-making in large-scale systems.


\bibliographystyle{apalike}
\bibliography{references}

\newpage
\appendix
\setcounter{equation}{0}
\renewcommand{\theequation}{\thesection.\arabic{equation}}
\renewcommand{\thefigure}{\thesection.\arabic{figure}} 
\setcounter{figure}{0} 
\setcounter{table}{0}
\renewcommand{\thetable}{\thesection.\arabic{table}}

\section{Algorithm Pseudocodes}

\begin{algorithm}[ht]\label{alg:rq-fpi-c}
\caption{Risk-Averse Quantal Fixed-Point Iteration (\textbf{RQ-FPI})}
\begin{algorithmic}

    \State \textbf{Input:} RQ-MFG $ \langle \X, \U, T, f, r,  \nu,\mathbb{M}, \Gamma^*_\mathbb{M},  D \rangle$, $\alpha$, $\tau$
    \State \textbf{Initialization:} Policy $\pi^{j=0}$, 

    \For{$j = 0, 1,\ldots$} 
    \State Compute MF flows $\S^j_\mathbb{M}$ for $\mu_0\in \mathbb{M}$ generated by $\pi^{j}$
        \For{$t = T-1, \ldots, 0$} 
            \State Set $Q^{\pi^{j+1}}_{\mu, t}(x, u)$ according to~\eqref{standard-Q} for all $x\in\X$, $u\in\U$ and flows $\mu \in \S^j_\mathbb{M}$
            \For{$x\in\X$}
                \State Compute $\pi^{j+1}_t(\cdot | x)$ from~\eqref{eqn:BR-RA-opt}
                \State Set $V^{\pi^{j+1}}_{\mu, t}(x) = \sum_{u\in\U}\pi^{{j}+1}_t(u|x)Q^{\pi^{j+1}}_{\mu, t}(x, u)$ for all flows $\mu \in \S^j_\mathbb{M}$
            \EndFor
        \EndFor
    \EndFor  
    \State \textbf{Return: $\pi$}
\end{algorithmic}
\end{algorithm}

\begin{algorithm}[ht]\label{alg:rq-fictitious-play-c}
\caption{Risk-Averse Quantal Fictitious Play (\textbf{RQ-Fictitious Play})}
\begin{algorithmic}

    \State \textbf{Input:} RQ-MFG $ \langle \X, \U, T, f, r,  \nu,\mathbb{M}, \Gamma^*_\mathbb{M},  D \rangle$, $\alpha$, $\tau$, $\beta$
    \State \textbf{Initialization:} Policy $\pi$, $\bar\pi = 0$

    \For{$j = 0, 1,\ldots$} 
    \State $\bar\pi = \beta\bar\pi  + (1-\beta)\pi$
    \State Normalize $\bar\pi$
    \State $\pi \leftarrow \B^{\textrm{RQE}}_\opt\circ\B^{\textrm{RQE}}_\prop(\bar\pi) $
    \EndFor  
    \State \textbf{Return: $\bar\pi $}
\end{algorithmic}
\end{algorithm}
 
\section{Experiments}\label{appdx:experiments}

\subsection{Computation of \texorpdfstring{$\pi^{*}_{\textrm{avg}}$}{pi*avg}}
We compute $\pi^{*}_{\textrm{avg}}$ using:
\begin{align}\label{eq:pi-avg-opt}
     \pi^{*}_{\textrm{avg}} &=\arg\max_{\pi\in\Pi}\sum_{k=1}^{|\mathbb{M}|}\Gamma^*_\mathbb{M}(\mu^k_0)J_{\mu^k} \big(\pi\big) - \alpha\nu(\pi),
\end{align}
where $\mu_0^k \in \mathbb{M}$ and $J_{\mu^k}(\pi)$ is defined in~\eqref{eqn:performance-infinite}. 
Due to the linearity of the expectation operator in~\eqref{eq:pi-avg-opt}, the corresponding dynamic programming equations admit a particularly simple structure. 
Specifically, the state–action value function satisfies
\begin{align}\label{avg-Q}
    Q^\pi_{\mu, t}(x, u) &= \sum_{k=1}^{|\mathbb{M}|}\Gamma^*_\mathbb{M}(\mu^k_0) \Big(r_t(x, u, \mu^k_t) + \sum_{x'\in\X}f_t(x'|x, u, \mu^k_t)V^\pi_{\mu^k, t+1}(x')\Big) \nonumber\\&= \sum_{k=1}^{|\mathbb{M}|}\Gamma^*_\mathbb{M}(\mu^k_0)Q^\pi_{\mu^k, t}(x, u) ,
\end{align}
where $Q^\pi_{\mu, T-1}(x, u) \triangleq \sum_{k=1}^{|\mathbb{M}|}\Gamma^*_\mathbb{M}(\mu^k_0) (r_{T-1}(x, u, \mu_{T-1}))$ and $Q^\pi_{\mu^k, t}(x, u)$ and $ V^\pi_{\mu^k, t+1}(x')\Big)$ follow from~\eqref{standard-Q} and~\eqref{standard-V} respectively. 
The corresponding value function inherits the same linear structure:
\begin{align*}
   V^\pi_{\mu, t}(x) &= \expct{Q^\pi_{\mu, t}(x, u)} {u\sim\pi_t}\\& = \expct{\sum_{k=1}^{|\mathbb{M}|}\Gamma^*_\mathbb{M}(\mu^k_0)Q^\pi_{\mu^k, t}(x, u) } {u\sim\pi_t} \\&= \sum_{k=1}^{|\mathbb{M}|}\Gamma^*_\mathbb{M}(\mu^k_0)\expct{Q^\pi_{\mu^k, t}(x, u) } {u\sim\pi_t} \\&= \sum_{k=1}^{|\mathbb{M}|}\Gamma^*_\mathbb{M}(\mu^k_0)V^\pi_{\mu^k, t}(x).
\end{align*}

\subsection{Experimental Details}

\begin{description}[leftmargin=1.2cm, style=nextline]
    \item[\textbf{Beach Bar}]
    \[
    \mathbb{M} =
    \left\{
    \begin{aligned}
        &[0.3, 0.2, 0.1, 0.4],\;
        [0.3, 0.3, 0.1, 0.3],\;
        [0.2, 0.2, 0.1, 0.5],\\
        &[0.25, 0.1, 0.25, 0.4],\;
        [0, 0, 0, 1]
    \end{aligned}
    \right\},
    \qquad
    \Gamma^*_{\mathbb{M}} = [0.25, 0.25, 0.2, 0.2, 0.1].
    \]

    \item[\textbf{Conservative Treasure Hunting}]
    \[
    \mathbb{M} =
    \left\{
    \begin{aligned}
        &[0.40, 0.35, 0.25],\;
        [0.41, 0.34, 0.25],\;
        [0.39, 0.36, 0.25],\\
        &[0.41, 0.35, 0.24],\;
        [0.39, 0.35, 0.26],\;
        [0.40, 0.36, 0.24],\\
        &[0.40, 0.34, 0.26],\;
        [0.405, 0.355, 0.24],\;
        [0.39, 0.355, 0.255],\\
        &[0.405, 0.34, 0.255]
    \end{aligned}
    \right\},
    \quad
    \Gamma^*_{\mathbb{M}}(\mu_0) = 0.1,\;\forall \mu_0 \in \mathbb{M}.
    \]

    \item[\textbf{Linear Quadratic}]
    \[
    \mathbb{M} =
    \left\{
    \begin{aligned}
        &[0.09,\ldots,0.10],\;
        [0.19,0.09,\ldots,0.00],\\
        &[0.09,0.18,0.09,\ldots,0.01,0.09]
    \end{aligned}
    \right\},
    \qquad
    \Gamma^*_{\mathbb{M}} = [0.5, 0.3, 0.2].
    \]

    \item[\textbf{Random Linear}]
    \[
    \mathbb{M} =
    \left\{
    \begin{aligned}
        &[0.2,0.2,0.2,0.2,0.2],\;
        [0.8,0.05,0.05,0.05,0.05],\\
        &[0.05,0.8,0.05,0.05,0.05],\;
        [0.5,0.5,0,0,0],\\
        &[0.4,0.3,0.3,0,0],\;
        [0.1,0.2,0.3,0.2,0.2],\\
        &[0.4,0.25,0.15,0.1,0.1],\;
        [0,0.6,0.2,0.1,0.1]
    \end{aligned}
    \right\},
    \quad
    \Gamma^*_{\mathbb{M}} =
    [0.12,0.08,0.15,0.10,0.20,0.05,0.18,0.12].
    \]
    
    \item[\textbf{Rock--Paper--Scissors}]
    \[
    \mathbb{M} =
    \left\{
    \begin{aligned}
    &[1,0,0,0],\;
    [0.9,0.1,0,0],\;
    [0.8,0.1,0.1,0],\\
    &[0.7,0.1,0.1,0.1],\;
    [0.4,0.3,0.2,0.1]
    \end{aligned}
    \right\},
    \quad
    \Gamma^*_{\mathbb{M}} = [0.4,0.25,0.15,0.1,0.1].
    \]

\end{description}

\begin{table}[ht]
\centering
\caption{Tuning Parameters used for RQ-FPI with Entropy Regularization}
\begin{tabular}{lcccc} 
 \toprule
 \textbf{Game}  & Horizon $T$ & $\alpha$ & $\tau$ & Number of fixed point iterations\\ 
 \midrule
SIS   & 50 & 5 & 0.2 & 25 \\ 
 Congestion   & 5 & 15 & 0.0667 & 30 \\ 
 Beach Bar   & 2 & 2 & 0.5 & 15 \\ 
 Conservative Treasure Hunting   & 5 & 1.5 & 0.6667 & 10 \\ 
 Linear Quadratic   & 3 & 1 & 1 & 10 \\ 
 Random Linear   & 3 & 28 & 0.0357 & 10 \\ 
  Rock-Paper-Scissors  & 7 & 10 & 0.1 & 15 \\ 
 \bottomrule
\end{tabular}
\label{table:fpi-params}
\end{table}

\begin{table}[ht]
\centering
\caption{Tuning Parameters used for RQ-FPI with Log-Barrier Regularization}
\begin{tabular}{lcccc} 
 \toprule
 Game  & Horizon $T$ & $\alpha$ & $\tau$ & Number of fixed point iterations\\ 
 \midrule
SIS   & 50 & 5 & 0.2 & 25 \\ 
 Congestion   & 5 & 6 & 0.1667 & 30 \\ 
 Beach Bar   & 2 & 2 & 0.5 & 15 \\ 
 Conservative Treasure Hunting   & 5 &  4& 0.25 & 10 \\ 
  Linear Quadratic   & 3 &  1& 1 & 10 \\ 
 Random Linear   & 3 & 40 & 0.025 & 10 \\ 
 Rock-Paper-Scissors  & 7 & 10 & 0.1 & 15 \\ 
 \bottomrule
\end{tabular}
\label{table:fpi-params-barrier}
\end{table}

\subsection{Additional Results}
Table~\ref{table:fpi-params} reports the hyperparameters used for RQ-FPI across all environments, and Figure~\ref{fig:convergence} from the main text illustrates the convergence behaviors of RQ-FPI and RQ-Fictitious Play under entropy regularization for each of these environments.
Tables~\ref{table:beach-bar-fpi-Results}--\ref{table:rps-fpi-Results} summarize the corresponding results, which are identical for RQ-FPI and RQ–Fictitious Play, from which we observe that the empirical trends identified in the SIS and congestion games persist across these settings.

Notably, with the exception of the Conservative Treasure Hunting and Rock–Paper–Scissors environments, the policy $\pi^*_{\textrm{avg}}$ is exploitable under uncertainty in the initial mean-field distribution, highlighting the necessity of explicitly modeling distributional risk.
In the Rock–Paper–Scissors environment, the optimal policy is invariant to the choice of initial mean-field distribution; consequently, the MF-RQE policy $\pi^{*}_{\textrm{RQE}}$ coincides with the standard entropy-regularized Nash equilibrium policy~\citep{cui2021approximately}, and all exploitability values are identically zero.

As mentioned in the main text, we extended our analysis beyond entropy regularization by considering the log-barrier regularizer, defined as $\nu(\pi(\cdot | x)) = -\sum_{u\in\U}\log(\pi(u | x))$.
Since the framework of \citet{cui2021approximately} is specifically tailored to entropy-regularized policies, a direct comparison with entropy-regularized Nash equilibrium policies under a single initial mean-field distribution $\mu_0 \in \mathbb{M}$ is not applicable in this setting.
Consequently, we restrict our comparison to the risk-neutral policy $\pi^{*}_{\textrm{avg}}$, which is computed using the same log-barrier regularization.
The resulting exploitability and empirical return statistics are summarized in Table~\ref{table:fpi-barrier-Results} in Appendix~\ref{appdx:experiments}.

\begin{table}[ht]
\centering
\caption{Computing $\Delta c(\pi)$ and Empirical Returns for Beach Bar}
\begin{tabular}{ccc} 
 \toprule
 Policy  & $\Delta c(\pi)$ & Expected Returns\\ 
 \midrule
$\pi^{*}_{\mu^1_0}$   & 0.057 & $0.063 \pm 0.017$ \\ 
 
 $\pi^{*}_{\mu^2_0}$   &  0.152 & $0.059 \pm 0.017$\\ 

 $\pi^{*}_{\mu^3_0}$   & 0.104 & $0.060 \pm 0.016$ \\

 $\pi^{*}_{\mu^4_0}$   & 0.081 & $0.059 \pm 0.016$\\

 $\pi^{*}_{\mu^5_0}$   & 4.770 & $-0.245 \pm 0.013$\\

 $\pi^{*}_{\textrm{avg}}$   & 1.4511 & $-0.052 \pm 0.015$\\

 $\pi^{*}_{\textrm{RQE}}$   & $\mathbf{0.00}$ & $\mathbf{0.065 \pm 0.016}$\\
 \bottomrule
\end{tabular}
\label{table:beach-bar-fpi-Results}
\end{table}

\begin{table}[ht]
\centering
\caption{Computing $\Delta c(\pi)$ and Empirical Returns for Conservative Treasure Hunting}
\begin{tabular}{ccc} 
 \toprule
 Policy  & $\Delta c(\pi)$ & Expected Returns\\ 
 \midrule
$\pi^{*}_{\mu^1_0},\pi^{*}_{\mu^2_0},\ldots,\pi^{*}_{\mu^{10}_0}$   & 0.037 & $0.759 \pm 0.002$ \\ 
 $\pi^{*}_{\textrm{RQE}} = \pi^{*}_{\textrm{avg}}$   & $\mathbf{0.00}$ & $\mathbf{0.763 \pm 0.001}$\\
 \bottomrule
\end{tabular}
\label{table:conservative-treasure-hunting-fpi-Results}
\end{table}

  \begin{table}[ht]
\centering
\caption{Computing $\Delta c(\pi)$ and Empirical Returns for Linear Quadratic}
\begin{tabular}{ccc} 
 \toprule
 Policy  & $\Delta c(\pi)$ & Expected Returns\\ 
 \midrule
$\pi^{*}_{\mu^1_0}$   & 0.334 & $-0.928 \pm 0.004$ \\ 
 $\pi^{*}_{\mu^2_0}$   &  0.317 & $  -0.929 \pm 0.004$\\ 

 $\pi^{*}_{\mu^3_0}$   & 0.307 & $ -0.929 \pm 0.004$ \\

 $\pi^{*}_{\textrm{avg}}$   & 0.066 & $-0.929 \pm 0.004$\\
 $\pi^{*}_{\textrm{RQE}}$   & $\mathbf{0.00}$ & $\mathbf{-0.928 \pm 0.004}$\\
 \bottomrule
\end{tabular}
\label{table:linear-quadratic-fpi-Results}
\end{table}

  \begin{table}[ht]
\centering
\caption{Computing $\Delta c(\pi)$ and Empirical Returns for Random Linear}
\begin{tabular}{ccc} 
 \toprule
 Policy  & $\Delta c(\pi)$ & Expected Returns\\ 
 \midrule
$\pi^{*}_{\mu^1_0}$   & 0.098 & $-0.601 \pm 0.127$ \\ 
 $\pi^{*}_{\mu^2_0}$   &  0.051 & $ -0.630 \pm 0.123$\\ 

 $\pi^{*}_{\mu^3_0}$   & 0.105 & $ -0.602\pm0.110$ \\

  $\pi^{*}_{\mu^4_0}$   & 0.070 & $-0.623 \pm0.113$ \\

  $\pi^{*}_{\mu^5_0}$   & 0.059 & $-0.599 \pm 0.116$ \\

  $\pi^{*}_{\mu^6_0}$   & 0.118 & $-0.591 \pm 0.143$ \\

  $\pi^{*}_{\mu^7_0}$   & 0.049 & $-0.608 \pm 0.124$ \\

  $\pi^{*}_{\mu^8_0}$   & 0.055 & $\mathbf{-0.588 \pm 0.116}$ \\

 $\pi^{*}_{\textrm{avg}}$   & 0.001 & $-0.613 \pm 0.131$\\

 $\pi^{*}_{\textrm{RQE}}$   & $\mathbf{0.00}$ & $-0.612 \pm 0.129$\\
 \bottomrule
\end{tabular}
\label{table:random-linear-fpi-Results}
\end{table}

\begin{table}[ht]
\centering
\caption{Computing $\Delta c(\pi)$ and Empirical Returns for Rock-Paper-Scissors}
\begin{tabular}{ccc} 
 \toprule
 Policy  & $\Delta c(\pi)$ & Expected Returns\\ 
 \midrule
$\pi^{*}_{\mu^1_0}, \pi^{*}_{\mu^2_0}, \ldots, \pi^{*}_{\mu^5_0}$   & $5.2e-7 \approx 0.00$ & $4.035 \pm 0.002$ \\ 
 $\pi^{*}_{\textrm{RQE}}= \pi^{*}_{\textrm{avg}}$   & $\mathbf{0.00}$ & $\mathbf{4.035 \pm 0.002}$\\
 \bottomrule
\end{tabular}
\label{table:rps-fpi-Results}
\end{table}

\begin{table}[t]
\centering
\caption{Exploitability $\Delta c(\pi)$ and empirical returns under RQ-FPI with log-barrier regularization.}
\label{table:fpi-barrier-Results}
\begin{tabular}{lcccc}
\toprule
\multicolumn{1}{c}{\textbf{Game}} &
\multicolumn{2}{c}{\textbf{$\pi^{*}_{\textrm{avg}}$}} &
\multicolumn{2}{c}{\textbf{$\pi^{*}_{\textrm{RQE}}$}} \\
\cmidrule(lr){2-3} \cmidrule(lr){4-5}
 & $\Delta c(\pi)$ & Return
 & $\Delta c(\pi)$ & Return \\
\midrule
SIS & 0.016 & $-22.218 \pm 0.062$ & \textbf{0.00} & $-22.232 \pm 0.062$ \\
Congestion & 0.003 & $-3.295 \pm 0.008$ & \textbf{0.00} & $-3.294 \pm 0.008$ \\
Beach Bar & 0.497 & $-0.037 \pm 0.016$ & \textbf{0.00} & $0.010 \pm 0.017$ \\
Conservative Treasure Hunting & 0.00 & $0.614 \pm 0.003$ & \textbf{0.00} & $0.614 \pm 0.003$ \\
Linear Quadratic & 0.013 & $-0.950 \pm 0.004$ & \textbf{0.00} & $-0.950 \pm 0.004$ \\
Random Linear & 0.00 & $-4.318 \pm 0.173$ & \textbf{0.00} & $-4.318 \pm 0.173$ \\
Rock--Paper--Scissors & 0.00 & $4.037 \pm 0.002$ & \textbf{0.00} & $4.037 \pm 0.002$ \\
\bottomrule
\end{tabular}
\end{table}

\section{Proof of Proposition~\ref{prop:existence-MF-RQE}}

\begin{lemma}\label{lmm:convexity}
Consider a fixed set of mean-field flows $\S_\mathbb{M}$.
Let $\pi_a, \pi_b \in \B^{\textrm{RQE}}_\opt(\S_\mathbb{M})$.
If there exists a $t<T$ such that $Q^{\pi^a}_{\mu^k, t}(x, u) = Q^{\pi^b}_{\mu^k, t}(x, u) \triangleq  Q_{k, t}(x, u)$ for all states $x\in\X$, $u\in\U$ and mean-field flows $\mu^k\in\S_\mathbb{M}$, then the following are true:
\begin{enumerate}[label=\roman*)]
    \item $\pi^a_t(\cdot |x) = \pi^b_t(\cdot|x)$
    \item $Q^{\pi^a}_{\mu^k, t-1}(x, u) = Q^{\pi^b}_{\mu^k, t-1}(x, u)$, 
\end{enumerate}
for all states $x\in\X$, $u\in\U$ and mean-field flows $\mu^k\in\S_\mathbb{M}$.
\end{lemma}

\begin{proof}
Under operator $\B^{\textrm{RQE}}_\opt$ it follows from~\eqref{eqn:BR-RA-opt} at each state $x\in \X$ and time step $t$ that,

\begin{align*}
\pi^a_t(\cdot | x) = \arg\min_{\pi'_t\in\Pi_t} c_t^{(\pi'_t, \pi^a_{-t}), ~\alpha}(x; \S_\mathbb{M})\quad\textrm{and}\quad
\pi^b_t(\cdot | x) = \arg\min_{\pi'_t\in\Pi_t} c_t^{(\pi'_t, \pi^b_{-t}), ~\alpha}(x; \S_\mathbb{M}),
\end{align*}
where,
\begin{align}\label{eq:complete-optimization}
    c_t^{(\pi'_t, \pi_{-t}), \alpha}(x; \S_\mathbb{M}) = \frac{1}{\tau}\log \Big(\sum_{k=1}^{|\mathbb{M}|}\Gamma^*_\mathbb{M}(\mu_0^k)\exp^{-\tau\pi'{_t}\t(\cdot | x)Q^{\pi}_{\mu^k, t}(x, \cdot)}\Big) + \alpha\nu(\pi'_t(\cdot | x)).
\end{align}

Note that the term $\pi_{-t}$ in~\eqref{eq:complete-optimization} enters $c_t^{(\pi'_t, \pi_{-t}), \alpha}(x; \S_\mathbb{M}) $ only through $Q^{\pi}_{\mu^k, t}(x, \cdot)$.
Since we have  $Q^{\pi^a}_{\mu^k, t}(x, u) = Q^{\pi^b}_{\mu^k, t}(x, u) \triangleq  Q_{k, t}(x, u)$ for all states $x\in\X$, $u\in\U$ and mean-field flows $\mu^k\in\S_\mathbb{M}$, it follows that,
\begin{align*}
    c_t^{(\pi'_t, \pi^a_{-t}), ~\alpha}(x; \S_\mathbb{M}) = c_t^{(\pi'_t, \pi^b_{-t}), ~\alpha}(x; \S_\mathbb{M}) = \frac{1}{\tau}\underbrace{\log \Big(\sum_{k=1}^{|\mathbb{M}|}\Gamma^*_\mathbb{M}(\mu_0^k)\exp^{-\tau\pi'{_t}\t(\cdot | x)Q_{k, t}(x, \cdot)}\Big)}_{\textrm{Convex in $\pi'_t$}} + \underbrace{\alpha\nu(\pi'_t(\cdot | x))}_{\textrm{Strictly Convex in $\pi'_t$}}.
\end{align*}
This implies that the objective function is strictly convex (sum of convex and strictly convex function) and must have a unique solution.
It follows that $\pi^a_t(\cdot |x) = \pi^b_t(\cdot|x)$ for all $x\in\X$.

Now, using \eqref{standard-Q}, for all states $x\in\X$, $u\in\U$ and mean-field flows $\mu^k\in\S_\mathbb{M}$, we have that
\begin{align}
    Q^{\pi^a}_{\mu^k, t-1}(x, u) &= r_{t-1}(x, u, \mu^k_{t-1}) + \sum_{x'\in\X}f_{t-1}(x'|x, u, \mu^k_{t-1})V^{\pi^a}_{\mu^k, t}(x')\\
    &= r_{t-1}(x, u, \mu^k_{t-1}) + \sum_{x'\in\X}f_{t-1}(x'|x, u, \mu^k_{t-1})\sum_{u\in\U}\pi^a_{t}(\cdot | x)Q^{\pi^a}_{\mu^k, t}(x', u)\\
    &= r_{t-1}(x, u, \mu^k_{t-1}) + \sum_{x'\in\X}f_{t-1}(x'|x, u, \mu^k_{t-1})\sum_{u\in\U}\pi^b_{t}(\cdot | x)Q^{\pi^b}_{\mu^k, t}(x', u)\\
    &= Q^{\pi^b}_{\mu^k, t-1}(x, u).
\end{align}
\end{proof}

\begin{lemma}\label{lmm:continuity-Q}
    The Q-function $Q^\pi_{\mu, t}(x, u)$ defined in~\eqref{standard-Q} is continuous in both $\mu$ and $\pi$ for all states $x\in\X$, $u\in\U$ and $t=0,\ldots, T-1$, where $d_\Pi(\pi, \pi')\triangleq \max_{x\in\X, t<T}\dtv{\pi_t(\cdot |x), \pi'_t(\cdot |x)}$ and $d_\M(\mu, \mu')\triangleq\max_{t<T}\dtv{\mu_t, \mu'_t}$.
\end{lemma}

\begin{proof}
    Continuity in $\mu$ follows directly from Assumption~\ref{assmpt:lipschitiz-model} and the fact that products and sums of Lipschitz and bounded functions are again Lipschitz and bounded~\citep{cui2021approximately}.
    Consider now a given mean-field flow $\mu$ and two policies $\pi^a, \pi^b\in\Pi$.
    We prove continuity with respect to the policies via induction.
    At $t=T-1$, for all $x\in\X$ and $u\in\U$,  $Q^{\pi^a}_{\mu, T-1}(x, u) = Q^{\pi^b}_{\mu, T-1}(x, u) = r_{T-1}(x, u, \mu_{T-1})$ and continuity follows trivially.
    Assume that at time $t+1$, $\vert Q^{\pi^a}_{\mu, t+1}(x, u) - Q^{\pi^b}_{\mu, t+1}(x, u)\vert \leq C_{t+1} d_\Pi(\pi^a, \pi^b)$, where $d_\Pi(\pi^a, \pi^b)\triangleq \max_{x\in\X, ~t < T}\dtv{\pi_t^a(\cdot|x),\pi_t^b(\cdot|x)}$ and $C_{t+1} >0$.
    Furthermore, under finite horizons and bounded rewards, we can uniformly bound $|Q^\pi_{\mu, t}(x, u)| \leq TB_r \triangleq B_Q$, where $B_r$ is the uniform bound for the bounded rewards. 
    Then, at time $t$, it follows that:
    
    \begin{align*}
        \vert Q^{\pi^a}_{\mu, t}(x, u) - Q^{\pi^b}_{\mu, t}(x, u)\vert &= | \sum_{x'\in\X}f_t(x'|x, u, \mu_t)\big(V^{\pi^a}_{\mu, t+1}(x') - V^{\pi^b}_{\mu, t+1}(x')\big)\vert\\
        &= | \sum_{x'\in\X}f_t(x'|x, u, \mu_t)\sum_{u'\in \U}\big(\pi^a_{t+1}(u'|x')Q^{\pi^a}_{\mu, t+1}(x') - \pi^b_{t+1}(u'|x')Q^{\pi^b}_{\mu, t+1}(x')\big)\vert\\
        &\leq  \sum_{x'\in\X}\sum_{u'\in \U}f_t(x'|x, u, \mu_t)\Big(\big\vert \pi^a_{t+1}(u'|x')-\pi^b_{t+1}(u'|x')||Q^{\pi^a}_{\mu, t+1}(x')|  \\&+\pi^b_{t+1}(u'|x')|Q^{\pi^a}_{\mu, t+1}(x')-Q^{\pi^b}_{\mu, t+1}(x')|\Big)\\
        &\leq \sum_{x'\in\X}f_t(x'|x, u, \mu_t)|Q^{\pi^a}_{\mu, t+1}(x')| \sum_{u'\in \U}\big\vert\pi^a_{t+1}(u'|x')-\pi^b_{t+1}(u'|x')| + C_{t+1}d_\Pi(\pi^a, \pi^b)\\
        &\leq \sum_{x'\in\X}f_t(x'|x, u, \mu_t)|Q^{\pi^a}_{\mu, t+1}(x')| d_\Pi(\pi^a, \pi^b) + C_{t+1}d_\Pi(\pi^a, \pi^b)\\
        &\leq (B_Q + C_{t+1}) d_\Pi(\pi^a, \pi^b) \triangleq C_td_\Pi(\pi^a, \pi^b).
    \end{align*}
By defining $C=\max_{t=0,\ldots, T-1}C_t$, we further obtain a uniform bound.
\end{proof}

\begin{corollary}\label{cor:continuity-cost}
   The regularized risk-averse cost~\eqref{eq: BR-RA-V} (expanded in~\eqref{eq:complete-optimization}) is continuous in $\pi'_t, \pi_t$ and $\S_{\mathbb{M}}$, where $d_\S(\S^1_\mathbb{M}, \S^2_\mathbb{M}) \triangleq \max_{k \in\{1, \ldots, |\mathbb{M}|\}, t< T}\dtv{\mu^{1,k}_t, \mu^{2,k}_t}$ where $\mu^{1,k}\in \S^1_\mathbb{M}$ and $\mu^{2,k}\in \S^2_\mathbb{M}$.
\end{corollary}
\begin{proof}
The log-sum-exp function is continuous, and the inputs to the exponents consist of terms linear (therefore continuous) in $\pi'_t$ and continuous in $Q^{\pi}_{\mu^k, t}(x, \cdot)$ for all states $x\in\X$ and mean-field flows $\mu\in\S_\mathbb{M}$ (from Lemma~\ref{lmm:continuity-Q}).
Furthermore, the convex regularizer is also continuous.
Thus, continuity follows directly from the fact that compositions and sums of continuous functions are continuous. 
\end{proof}

\begin{proof}
    We define the operator $\Phi = \B^{\textrm{RQE}}_\opt\circ\B^{\textrm{RQE}}_\prop: \Pi\to\Pi$ and show that the requirements of Kakutani's fixed point theorem hold for $\Phi$.

    The space of policies $\Pi$ can be identified as the 
    Cartesian product of \textit{finite} simplexes $ \prod_{t=0}^{T-1}\prod_{x\in \X}\P({\U)}$ where $\pi_t(\cdot | x)  \subseteq \P({\U})$ for all $x\in \X$ and $t\in\{0, \ldots, T-1\}$.
    The finite-dimensional simplices are non-empty, convex and compact subsets of the Euclidean space, and hence their finite Cartesian product $\Pi$ have the same properties.

    Note that $\Phi$ maps to a non-empty set as the set of induced mean-field flows are uniquely determined and the entire domain $\Pi$ of the operator $\B^{\textrm{RQE}}_\opt$ forms the feasible set.
    For any $\pi, \Phi(\pi)$ is trivially convex, since the optimal policy is unique as shown in the following. 
    Consider $\pi_a, \pi_b \in \Phi(\pi)$. For $t=T-1$, from~\eqref{standard-Q}, we have $Q^{\pi^a}_{\mu^k, T-1}(x, u) = Q^{\pi^b}_{\mu^k, T-1}(x, u) \triangleq  r_{T-1}(x, u, \mu^k_{T-1})$ for all states $x\in\X$, $u\in\U$ and mean-field flows $\mu^k\in\B^{\textrm{RQE}}_\prop(\pi)$.
    From Lemma~\ref{lmm:convexity}, it recursively follows that $\pi^a_t(u|x) = \pi^b_t(u|x)$ for all $x\in\X$, $u\in\U$ and $t=\{0,\ldots, T-1\}$.
    Consequently, $\pi^a = \pi^b$ and the optimal policy is unique.

    Lastly, we show that $\Phi$ has a closed graph.
    Let $(\pi_n, \pi'_n)\to (\pi, \pi')$ be arbitrary sequences such that $\pi'_n \in \Phi(\pi_n)$.
    It is sufficient to show that $\pi'\in\Phi(\pi)$~\cite{cui2021approximately}.
Let us assume that that $\pi'\notin\Phi(\pi)$.
This implies that for some $x\in\X$ and $t\in\{0, \ldots, T-1\}$, there exists $\bar\pi_t(\cdot|x)\in\Pi_t$ such that
\begin{align*}
    c_t^{(\bar\pi_t, \pi'_{-t}), ~\alpha}(x; \B^{\textrm{RQE}}_\prop(\pi)) < c_t^{(\pi'_t, \pi'_{-t}), ~\alpha}(x; \B^{\textrm{RQE}}_\prop(\pi)).
\end{align*}
    Define
    \begin{align}\label{eq:delta-smol}
        \delta \triangleq   c_t^{(\pi'_t, \pi'_{-t}), ~\alpha}(x; \B^{\textrm{RQE}}_\prop(\pi)) - c_t^{(\bar\pi_t, \pi'_{-t}), ~\alpha}(x; \B^{\textrm{RQE}}_\prop(\pi)).
    \end{align}
  From Corollary~\ref{cor:continuity-cost}, there exists $N_1, N_2\in \mathbb{N}$ such that,
    \begin{align}\label{eq:delta-N1-N2}
    \vert c_t^{(\bar\pi_t, \pi'_{n, -t}), ~\alpha}(x; \B^{\textrm{RQE}}_\prop(\pi_n)) -  c_t^{(\bar\pi_t, \pi'_{-t}), ~\alpha}(x; \B^{\textrm{RQE}}_\prop(\pi))
\vert < \frac{\delta}{4}\quad \forall ~n > N_1, \nonumber\\
\vert c_t^{(\pi'_t, \pi'_{n, -t}), ~\alpha}(x; \B^{\textrm{RQE}}_\prop(\pi_n)) -  c_t^{(\pi'_t, \pi'_{-t}), ~\alpha}(x; \B^{\textrm{RQE}}_\prop(\pi))
\vert < \frac{\delta}{4}\quad \forall ~n > N_2.
    \end{align}
    
    Similarly, there exists $M\in \mathbb{N}$ such that for all $n>M$,
    \begin{align}\label{eq:delta-M}
    \vert c_t^{(\pi'_{n, t}, \pi'_{n, -t}), ~\alpha}(x; \B^{\textrm{RQE}}_\prop(\pi_n)) -  c_t^{(\pi'_{t}, \pi'_{n, -t}), ~\alpha}(x; \B^{\textrm{RQE}}_\prop(\pi_n))
\vert < \frac{\delta}{4}\quad \forall ~n > M.
    \end{align}

    It follows from~\eqref{eq:delta-smol}-~\eqref{eq:delta-M} that for all $n>\max(N_1, N_2, M)$,
    \begin{align*}
        c_t^{(\bar\pi_t, \pi'_{n, -t}), ~\alpha}(x; \B^{\textrm{RQE}}_\prop(\pi_n)) < c_t^{(\pi'_{n, t}, \pi'_{n, -t}), ~\alpha}(x; \B^{\textrm{RQE}}_\prop(\pi_n)),
    \end{align*}
    thereby implying that $\pi_n' \notin\Phi(\pi_n)$, which is a contradiction.
    It follows that $\Phi$ has a closed graph.
    By Kakutani's fixed point theorem, there exists a fixed point $\pi^*_{\textrm{RQE}} $ that generates the set of mean-field flows $\B^{\textrm{RQE}}_\prop(\pi^*_{\textrm{RQE}})$.
    The pair $(\pi^*_{\textrm{RQE}} ,\B^{\textrm{RQE}}_\prop(\pi^*_{\textrm{RQE}})) $ is an MF-RQE by definition.
 
\end{proof}

\section{Proof of Theorem~\ref{thm:MF-RQE-finite-population}}

\begin{lemma}\label{lmm:pre-finite-dev}
For all $t\geq t'$, the empirical distributions $\M^{N_1}_{t}, \hat\M^{N}_{t}$ under the identical policy $\pi$ obey
    \begin{align*}
    \mathbb{E}_{{\pi}}
    \Big[ \dtv{\M^{N_1}_{t}, \hat\M^{N}_{t}}  \Big \vert \bfx_{t'}^{N_1}=\hat\bfx_{t'}^{N_1}\! \Big] \leq \Residue.
\end{align*}
\end{lemma}

\begin{proof}
The proof follows a similar structure to Proposition 3 in~\cite{jeloka2025learninglargescalecompetitiveteam}. We begin by computing:
\begin{align}\label{eq:blue-splits-estimator-eval}
    \mathbb{E}_{{\pi}}
    \Big[ \dtv{\M^{N}_{t+1}, \hat\M^{N}_{t+1}}   \Big  \vert \bfx_t, \hat\bfx_t\Big]  
    &\leq\mathbb{E}_{{\pi}}
    \Big[ \dtv{\M^{N}_{t+1}, \M_{t+1}} \Big \vert \bfx_t, \hat\bfx_t \Big]\nonumber\\
    &+ \mathbb{E}_{{\pi}}
    \Big[ \dtv{\M_{t+1}, \hat\M_{t+1}} \Big \vert \bfx_t, \hat\bfx_t \Big]\nonumber\\
    &+ \mathbb{E}_{{\pi}}
    \Big[ \dtv{\hat\M^{N}_{t+1}, \hat\M_{t+1}} \Big \vert \bfx_t, \hat\bfx_t \Big],
\end{align}
where $\M_{t+1} = \M^{N}_{t}F(\M^{N}_{t}, \pi)$ and $\hat\M_{t+1} = \hat\M^{N}_{t}F(\hat\M^{N}_{t}, \pi)$ are shorthand notations for the next induced mean-field from the \textit{infinite}-population \textit{deterministic} dynamics~\citep{guan2024zero}. 
From~\cite{guan2024zero}, we have 
\begin{align*}
    \mathbb{E}_{{\pi}}
     \Big[ \dtv{\M^{N}_{t+1}, \hat\M^{N}_{t+1}}  \Big \vert \bfx_t, \hat\bfx_t\Big]  \leq \Residue,~\textrm{and}~~
     \mathbb{E}_{{\pi}}
    \Big[ \dtv{\hat\M^{N}_{t+1}, \hat\M_{t+1}} \Big \vert \bfx_t, \hat\bfx_t \Big]\leq \Residue.
\end{align*}
Thus, we are left to simplify 
$\mathbb{E}_{{\pi}}\Big[ \dtv{\M_{t+1}, \hat\M_{t+1}} \Big \vert \bfx_t, \hat\bfx_t\Big] = \dtv{\M^{N}_{t}F(\M^{N}_{t}, \pi), \hat\M^{N}_{t}F(\hat\M^{N}_{t}, \pi)}$, 
where the expectation vanishes due to \textit{deterministic} transitions under identical policies $\pi$.
It follows, 
\begin{align}\label{eq:I1-I2}
    2\dtv{\M^{N}_{t}F(\M^{N}_{t}, \pi), &\hat\M^{N}_{t}F(\hat\M^{N}_{t}, \pi)} \\
    &=\sum_{x'\in\X}\Big\vert \sum_{x\in \X}\sum_{u \in \U} f_t(x'|x, u, \M^{N}_{t})\pi_t(u|x) \M^{N}_{t}(x) \nonumber \\
    &-\sum_{x \in \X} \sum_{u \in \U} f_t(x'|x, u,\hat\M^{N}_{t})\pi_t(u|x) \hat\M^{N}_{t}(x)\Big\vert \nonumber \\
    &\leq \sum_{x'\in\X}\sum_{x\in \X}\sum_{u \in \U} \Big\vert  f_t(x'|x, u, \M^{N}_{t})\pi_t(u|x) \M^{N}_{t}(x) \nonumber 
    - f_t(x'|x, u,\hat\M^{N}_{t})\pi_t(u|x) \hat\M^{N}_{t}(x)\Big\vert\nonumber .
\end{align}

We now individually bound the terms in using the definition of total variation distance, the identity
\begin{align*}
|b(ac -a'c')| \leq b(|a-a'|c  +a'|c-c'|), \quad a,b,c\geq 0,
\end{align*}
and Assumption~\ref{assmpt:lipschitiz-model} as follows:
\begin{align}\label{eq:blue-step-1-pre-final}
2\dtv{\M^{N}_{t}F(\M^{N}_{t}, \pi), \hat\M^{N}_{t}F(\hat\M^{N}_{t}, \pi)} \leq (2+ L_f)\dtv{\M^{N}_{t}, \hat \M^{N}_{t}}.
\end{align}

Thus, we can rewrite~\eqref{eq:blue-splits-estimator-eval} using the above computations to obtain
\begin{align}\label{eq:blue-step-1-final}
     &\mathbb{E}_{{\pi}}
    \Big[ \dtv{\M^{N}_{t+1}, \hat\M^{N}_{t+1}}  \Big \vert \bfx_t, \hat\bfx_t\Big]  \leq\frac{1}{2}(2+L_f)\dtv{\M^{N}_{t}, \hat \M^{N}_{t}} + \Residue.
\end{align}

Define $a_t \triangleq \mathbb{E}_{{\pi}}
    \Big[ \dtv{\M^{N_1}_{t}, \hat\M^{N}_{t}}  \Big \vert \bfx_{t'}^{N_1}=\hat\bfx_{t'}^{N_1}\! \Big]$.
At $t=t'$, $\M^{N_1}_{t'} = \hat\M^{N_1}_{t'} $.
Thus, $a_{t'} = 0$.
We proceed to show that $a_{t+1} \leq \kappa_1 + \kappa_2a_t$ where $\kappa_1=\Residue$ and $\kappa_2=\frac{1}{2}(2+L_f)$.
From the law of iterated expectations, 
\begin{align*}
    \mathbb{E}_{\pi}
    \Big[ \dtv{\M^{N_1}_{t+1}, &\hat\M^{N}_{t+1}}  \Big \vert \bfx_{t'}^{N_1}=\hat\bfx_{t'}^{N_1}\! \Big]
    \\&= \mathbb{E}_{{\pi}}
    \left[ \mathbb{E}_{{\pi}}\Big[\dtv{\M^{N}_{t+1}, \hat\M^{N}_{t+1}}  \Big \vert \bfx_t, \hat\bfx_t\Big]\Big \vert \bfx_{t'}^{N}=\hat\bfx_{t'}^{N} \right]\\
    &\leq \mathbb{E}_{{\pi}}
    \left[\kappa_1 + \kappa_2\dtv{\M^{N}_{t}, \hat\M^{N}_{t}} \vert \bfx_{t'}^{N}=\hat\bfx_{t'}^{N} \! \right]\\
    &= \kappa_1 + \kappa_2 a_t.
\end{align*}
We have $a_{t'} = 0$, $a_{t'+1} < \kappa_1$, $a_{t'+2} < \kappa_1 + \kappa_2a_{t'+1}$ or $a_{t'+2} < \kappa_1(\kappa_2 + 1)$, $a_{t'+3} < \kappa_2a_{t'+2} + \kappa_1$ or $a_{t'+3} < \kappa_1(\kappa_2^2 + \kappa_2 + 1) $ and so on.
This can be written compactly for all $t>t'$ as ($\kappa_2  \neq 1$)
\begin{align*}
    a_{t} \leq \kappa_1\sum_{s=1}^{t-t'}\kappa_2^s = \kappa_1 \frac{\kappa_2^{t-t'}-1}{\kappa_2-1} = \Residue.
\end{align*}
\end{proof}

\begin{lemma}\label{lmm:finite-deviation}
    Consider the sequence of empirical distributions $\mathcal{M}^N = \{\M^N_0, \ldots, \M^N_{t-1}, \bar\M^N_t, \bar{\bar\M}^N_{t+1}, \ldots \bar{\bar\M}^N_T\}$ generated by all agents following an identical policy $\pi$ from time steps $0$ to $t-1$, allowing one agent to deviate at time step $t$ and $t+1$, and then following $\pi$ again, where $\mathbf{X}_0\sim \mu_0$ and $\M_0 = \empMu{\mathbf{X}_0}$.
    Furthermore, define the deterministic mean-field flow induced by policy $\pi$ as $\mu = \{\mu_0, \mu_1, \ldots, \mu_T\}$. 
    Then, $\expct{\dtv{\M^N, \mu} \big\vert \mathbf{X}_0\sim \mu_0}{\pi} \leq \Residue$.
\end{lemma}
\begin{proof}
At time step $0$, it follows from the law of large numbers that $\expct{\dtv{\M^N_0, \mu_0}}{\pi} \leq \Residue$~\citep{guan2024zero}.
    For all time steps $1$ to $t-1$, starting from the identical initial distribution $\mu_0$ and following an identical policy $\pi$, from~\cite{guan2024zero}, we have 
\begin{align*}
    \mathbb{E}_{{\pi}}
    \Big[ \dtv{\M^{N}_{t'+1}, \M_{t'+1}} \Big \vert \bfx_{t'} \Big] &\leq \Residue.
\end{align*}
Then, by the law of iterated expectations, it follows that
\begin{align*}
    \expct{\dtv{\M^N_{t'+1}, \mu_{t'+1}} \big\vert \mathbf{X}_0\sim \mu_0}{\pi} 
    &= \mathbb{E}_{\pi}\Big[
    \expct{ \dtv{\M^{N}_{t'+1}, \M_{t'+1}} \Big \vert \mathbf{X}_{t'}}{\pi} \big\vert \mathbf{X}_0\sim \mu_0\Big]\leq \Residue.
\end{align*}

Without loss of generality, we can assume that the deviating agent $i=1$.
The deviated empirical distribution at time $t$ is given by 
\begin{align*}
    \bar\M^N_t = \frac{1}{N} \sum_{i=2}^{N_1} \mathds{1}_x(X_{i,t}) + \frac{1}{N} \mathds{1}_x(\bar{X}_{1,t}).
\end{align*}
Consequently, 
\begin{align}\label{eq:deviation}
    \dtv{\bar\M^{N}_{t}, \M^{N}_{t}} &= \big\vert\frac{1}{N} \sum_{i=2}^{N_1} \mathds{1}_x(X_{i,t}) + \frac{1}{N} \mathds{1}_x(\bar{X}_{1,t}) - \frac{1}{N} \sum_{i=1}^{N_1} \mathds{1}_x(X_{i,t})\big\vert\nonumber\\
    &= \big\vert \frac{1}{N} \mathds{1}_x(\bar{X}_{1,t})\big\vert \leq \frac{1}{N}.
\end{align}

\begin{align*}
    \expct{\dtv{\bar\M^N_{t}, \mu_{t}} \big\vert \mathbf{X}_0\sim \mu_0}{\pi} 
    &= \mathbb{E}_{\pi}\Big[
    \expct{ \dtv{\bar\M^{N}_{t}, \M_{t}} \Big \vert \mathbf{X}_{t-1}^{N}}{\pi} \big\vert \mathbf{X}_0\sim \mu_0\Big]\\
    &\leq \mathbb{E}_{\pi}\Big[
    \expct{ \dtv{\bar\M^{N}_{t}, \M^{N}_{t}}  + \dtv{\M^{N}_{t}, \M_{t}}\Big \vert \mathbf{X}_{t-1}^{N}}{\pi} \big\vert \mathbf{X}_0\sim \mu_0\Big]\\
    &\leq \Residue,
\end{align*}
where the last inequality comes by combining the above two results.
At each time step $s\geq t+1$,
\begin{align*}
    &\expct{\dtv{\bar{\bar\M}^N_{s+1},  \mu_{s+1}} \big\vert \mathbf{X}_0\sim \mu_0}{\pi} \\
  &\qquad\leq \mathbb{E}_{\pi}\Big[
    \expct{ \dtv{\bar{\bar\M}^{N}_{s}, \bar\M^{N}_{s}}  + \dtv{{\bar\M}^{N}_{s}, \M^{N}_{s}}  + \dtv{\M^{N}_{s}, \M_{s}}}{\pi}  \big\vert \mathbf{X}_0 \sim \mu_0\Big]\\
    & \qquad = \mathbb{E}_{\pi}\Big[
    \expct{ \expct{\dtv{\bar{\bar\M}^{N}_{s}, \bar\M^{N}_{s}} \Big \vert \mathbf{\bar{\bar{X}}}_{t}^{N}= \bar{\mathbf{X}}_{t}}{\pi}  + \dtv{{\bar\M}^{N}_{s}, \M^{N}_{s}}  \Big \vert \mathbf{X}_{t-1}^{N}=\mathbf{\bar{X}}_{t-1}^{N}}{\pi}  + \dtv{\M^{N}_{s}, \M_{s}} \big\vert \mathbf{X}_0 \sim \mu_0\Big]\\
    & \qquad\leq \Residue,
\end{align*}
where the last inequality follows from Lemma~\ref{lmm:pre-finite-dev}.
We have shown that $\dtv{\M^N_t, \mu_t}\leq \Residue$ for all $t\in\{0, \ldots, T-1\}$.
Thus, $\expct{\dtv{\M^N, \mu} \big\vert \mathbf{X}_0\sim \mu_0}{\pi} \triangleq \max_t\dtv{\M^N_t, \mu_t}\leq \Residue$.
\end{proof}

\begin{proof}[Proof of Theorem~~\ref{thm:MF-RQE-finite-population}]
We can now proceed to proving Theorem~\ref{thm:MF-RQE-finite-population}
by computing a relation between the finite-population objective
\begin{align}\label{eq:BR-RA-V-finite}
c_t^{\alpha, i, N}(\pi_1,\!\ldots\!,\pi_N; x)
\!=\!\rho_{\Gamma^*_\mathbb{M}}\left( V^\pi_{i,t}(x) \right)
\!+\! \alpha \nu\big(\pi_{i,t}(\cdot|x)\big),
\end{align}
and the infinite-population objective~\eqref{eq: BR-RA-V} under an MF-RQE $(\pi^{*}_{\textrm{RQE}},  \S^*_\mathbb{M})$.
In particular, we show that for all agents $i\in[N]$, for all states $x\in\X$ and times $t\in\{0,\ldots, T-1\}$,
\begin{align*}
      \big\vert c_t^{\pi^{*}_{\textrm{RQE}}, \alpha}(x; \S^*_\mathbb{M}) - c_t^{\alpha, i, N}\big(\pi^*_{\textrm{RQE} , i}, \pi^*_{\textrm{RQE}, -i}; ~x\big) \big\vert \leq \Residue,
\end{align*} 
and for all $\pi_t\in\Pi_t$,
\begin{align*}
      \big\vert c_t^{(\pi_t, \pi^{*}_{\textrm{RQE}, -t}), \alpha}(x; \S^*_\mathbb{M}) - c_t^{\alpha, i, N}\big((\pi_{t}, \pi^*_{\textrm{RQE}, i, -t}), \pi^*_{\textrm{RQE}, -i}; ~x\big) \big\vert \leq \Residue.
\end{align*} 
Using these two relations along with the fact that $c_t^{\pi^{*}_{\textrm{RQE}}, \alpha}(x; \S^*_\mathbb{M}) \leq c_t^{(\pi_t, \pi^{*}_{\textrm{RQE}, -t}), \alpha}(x; \S^*_\mathbb{M}) $ for all $\pi_t\in\Pi_t$, we can directly obtain that the identical MF-RQE policy $\pi^{*}_{\textrm{RQE}}$ is an $\epsilon$-RQE for the finite population game where $\epsilon = \Residue$.

We define the finite-population RQ-MFG value function for agent $i\in[N]$ at any given state $x\in\X$ and at time $t$ as follows:
\begin{align}\label{eq:finite-V}
V^\pi_{i,t}(x_{i}, \x_{-i, t}) =  \expct{{\sum_{t'=t}^Tr_t(x_{i,t'}, u_{i, t'}, \mu^N_{t'})} \big\vert x_{i}, \x_{-i, t}}{\pi},
\end{align}
where $\x_{-i, t} = \{x_{j, t}\}_{j\neq i}$ is the joint state of all the agents except $i$.
Note that the expectation is taken with respect to the joint policy $\pi$ and the stochastic transition dynamics.
Using~\eqref{eq:finite-V} and Theorem~\ref{thm:dual-representation} under the KL penalty function, we can expand~\eqref{eq:BR-RA-V-finite} for all states $x\in\X$ as follows:
\begin{align}\label{eq:dual-rep=finite}
     &\rho_{\Gamma^*_\mathbb{M}}\left( V^\pi_{i,t}(x) \right)\nonumber\\
     &\equiv \sup_{\hat{\Gamma}_\mathbb{M}\in\P(\mathbb{M})} -\expct{V^\pi_{i,t}(x, \x_{-i, t})}{\mu_0\sim\hat{\Gamma}_\mathbb{M}} -  \frac{1}{\tau}D(\hat{\Gamma}_\mathbb{M}, \Gamma^*_\mathbb{M})\nonumber\\
     &= \frac{1}{\tau}\log \Big(\sum_{k=1}^{|\mathbb{M}|}\Gamma^*_\mathbb{M}(\mu_0^k)\exp^{-\tau V^\pi_{i,t}(x, \x^k_{-i, t})}\Big)\nonumber\\
     &=\frac{1}{\tau}\textrm{LSE} \Big(\log w_k{-\tau \sum_{u\in\U}\pi_{i, t}(u|x)\Big(r(x, u, 
     \bar\mu^{N, k}_t) + \expct{{\sum_{t'=t+1}^Tr_t(x_{i,t'}, u_{i, t'}, \bar\mu^{N, k}_{t'})} \big\vert x_{i, t} = x, u_{i, t} =u,  \x^k_{-i, t}}{\pi}}\Big)\Big),
\end{align}
where $\x^k_{-i, t}$ is obtained by following joint policy $\pi$ starting at time $t=0$ and $\x^k_{-i, 0} \sim \mu^k_0 \in\mathbb{M}$, $\bar\mu^{N, k}_t =\empMu{\bar\x_t}$ where $\bar\x^k_t =(x, \x^k_{-i, t})$.
For agent $i$ at time $t$, we consider the arbitrary state $x$.
We further define $\mu^{N, k}_t = \empMu{\x^k_t}$, where all agents follow policy $\pi$ starting at time $t=0$ and $\x^k_{0} \sim \mu^k_0 \in\mathbb{M}$.

For brevity, we have defined $w_k = \Gamma^*_\mathbb{M}(\mu_0^k)$ and denote the log-sum-exp operator by $\textrm{LSE}$.
Intuitively, \eqref{eq:dual-rep=finite} formulates the risk-aversion for agent $i$ at each state $x$ for different trajectories of the joint state, arising due to the different realizations of the initial distribution.

Using similar notation, we can re-write~\eqref{eq:dual-rep} as
\begin{align}\label{eq:dual-rep-rewrite}
     \rho_{\Gamma^*_\mathbb{M}}\left( V^\pi_{\mu, t}(x) \right)
     &=\frac{1}{\tau}\textrm{LSE} \Big(\log w_k{-\tau \sum_{u\in\U}\pi_{t}(u|x)\Big(r(x, u, 
     \mu^{k}_t) + \expct{{\sum_{t'=t+1}^Tr_t(x_{t'}, u_{t'}, \mu^{k}_{t'})} \big\vert x_{t} = x, u_{t} =u}{\pi}}\Big)\Big),
\end{align}
where $\mu^k\in\S^*_\mathbb{M}$ constitute the MF-RQE set of flows.
Notice that the term $c_t^{\alpha, i, N}\big(\pi^*_{\textrm{RQE} , i}, \pi^*_{\textrm{RQE}, -i}; ~x\big)$ computes the mean-field flows under identical policy $\pi^*_{\textrm{RQE}}$, with agent $i$ is specifically considered to be at state $x$ at time $t$, following which, $\pi^*_{\textrm{RQE}}$ is again followed.
The term $c_t^{\alpha, i, N}\big((\pi_{t}, \pi^*_{\textrm{RQE}, i, -t}), \pi^*_{\textrm{RQE}, -i}; ~x\big)$ follows the same suit, with an additional deviation by agent $i$ from $x_t=x$ to $x_{t+1}$ since it follows the deviated policy $\pi_t$.
Thus, we can use Lemma~\ref{lmm:finite-deviation} for both these cases and show that the mean-field flows  $\expct{\dtv{\M^{N, k}, \mu^k} \big\vert \mathbf{X}^k_0\sim \mu^k_0}{\pi^*_{\textrm{RQE}}} \leq \Residue$ for all $k=1,\ldots, |\mathbb{M}|$.
Thus, $ \expct{\dtv{\S^N_\mathbb{M}, \S^*_\mathbb{M}}}{\pi^*_{\textrm{RQE}}} =\max_{k=1,\ldots, |\mathbb{M}|}\expct{\dtv{\M^{N, k}, \mu^k} \big\vert \mathbf{X}^k_0\sim \mu^k_0}{\pi^*_{\textrm{RQE}}} \leq \Residue$.
Further, by Corollary~\ref{cor:continuity-cost}, we have continuity of the cost function with respect to the mean-field flows.
Consequently, 
\begin{align*}
      \big\vert c_t^{\pi^{*}_{\textrm{RQE}}, \alpha}(x; \S^*_\mathbb{M}) - c_t^{\alpha, i, N}\big(\pi^*_{\textrm{RQE} , i}, \pi^*_{\textrm{RQE}, -i}; ~x\big) \big\vert \leq  \Residue,
\end{align*} 
and for all $\pi_t\in\Pi_t$
\begin{align*}
      \big\vert c_t^{(\pi_t, \pi^{*}_{\textrm{RQE}, -t}), \alpha}(x; \S^*_\mathbb{M}) - c_t^{\alpha, i, N}\big((\pi_{t}, \pi^*_{\textrm{RQE}, i, -t}), \pi^*_{\textrm{RQE}, -i}; ~x\big) \big\vert \leq \Residue,
\end{align*} 
and the result follows.

\end{proof}

\section{Proof of Theorem~\ref{thm:FPI-convergence}}
\begin{lemma}\label{lmm:prop-lipschitz}
    The operator $\B^{\textrm{RQE}}_\prop(\pi)$ is Lipschitz with some constant $K_\prop$.
\end{lemma}
\begin{proof}
    From Lemma B.7.4 in~\citep{cui2021approximately}, for any single initial distribution, i.e., $\B_\prop(\pi)$ is Lipschitz with constant $K_\prop$.
    This implies,
    \begin{align}\label{eq:single-flow-Lipschitz}
        d_\M(\B_\prop(\pi_a), \B_\prop(\pi_b)) \leq K_\prop d_\Pi(\pi_a, \pi_b).
    \end{align}
    Consequently, given a set of initial distributions $\mathbb{M}$, 
    \begin{align*}
        d_{\S}(\B^{\textrm{RQE}}_\prop(\pi_a),\B^{\textrm{RQE}}_\prop(\pi_b)) &=\max_{k \in\{1, \ldots, |\mathbb{M}|\}}\dtv{\mu^{a,k}, \mu^{b,k}} \\
        &= \max_{k \in\{1, \ldots, |\mathbb{M}|\}}d_\S(\B_\prop(\pi_a), \B_\prop(\pi_b))_{\mu_0 = \mu^k_0}\\
        &\leq K_\prop d_\Pi(\pi_a, \pi_b),
    \end{align*}
    as both $\mu^{a,k}, \mu^{b,k}$ originate from the same initial distribution $\mu^k_0$ and we can apply the result from~\eqref{eq:single-flow-Lipschitz}.
\end{proof}

\begin{lemma}\label{lmm:strong-convexity}
    Consider a function $f(\pi; \theta)$ that is $(\alpha\mu)$-strongly convex in $\pi$, where $\pi$ and $\theta$ are equipped with metrics $d_\Pi$ and $d_\theta$ respectively, such that it satisfies
    \begin{align}\label{eq:grad-lipschitz}
        \|\nabla_\pi f(\pi; \theta_a) - \nabla_\pi f(\pi; \theta_b)\| \leq \bar{C} d_\theta (\theta_a, \theta_b), \quad \bar C>0.
    \end{align}
    Further, define the (unique) optimizers $\pi^*_a = \arg\min f(\pi; \theta_a)$ and $\pi^*_b = \arg\min f(\pi; \theta_b)$.
    It follows that,
    \begin{align*}
        \dtv{\pi^*_a, \pi^*_b}\leq \frac{C}{\alpha\mu}d_\theta (\theta_a, \theta_b), \quad C>0.
    \end{align*} 
\end{lemma}

\begin{proof}
    Strong convexity implies that
    \begin{align}\label{eq:strong-grad-convexity}
        (\pi_1 - \pi_2)\t\big(\nabla_\pi f(\pi_1; \theta) - \nabla_\pi f(\pi_2; \theta)\big) \geq \alpha\mu \|\pi_1-\pi_2\|_2^2.
    \end{align}
    From the optimality {of $\pi^*_a$ and $\pi^*_b$} 
    we have $\nabla_\pi f(\pi^*_a; \theta_a)=0$ and $\nabla_\pi f(\pi^*_b; \theta_b)=0$, therefore $\nabla_\pi f(\pi^*_a; \theta_a) - \nabla_\pi f(\pi^*_b; \theta_b) = 0$.
    Adding and subtracting $\nabla_\pi f(\pi^*_b; \theta_a)$ gives
    \begin{align*}
        \underbrace{\nabla_\pi f(\pi^*_a; \theta_a) - \nabla_\pi f(\pi^*_b; \theta_a)}_{I_1}  + \underbrace{\nabla_\pi f(\pi^*_b; \theta_a) - \nabla_\pi f(\pi^*_b; \theta_b)}_{I_2} = 0.
    \end{align*}
    It follows that,
    \begin{align*}
        I_1 &= -I_2\\
        (\pi^*_a-\pi^*_b)\t I_1 &= -(\pi^*_a-\pi^*_b)\t I_2.
    \end{align*}
    Applying~\eqref{eq:strong-grad-convexity} on the LHS term gives $(\pi^*_a-\pi^*_b)\t I_1 \geq \alpha\mu d_\Pi(\pi^*_a, \pi^*_b)^2$.
    This implies $(\pi^*_a-\pi^*_b)\t I_1 = |(\pi^*_a-\pi^*_b)\t I_1| = |(\pi^*_a-\pi^*_b)\t I_2|$.
    Using Cauchy-Schwartz and~\eqref{eq:grad-lipschitz} for $|(\pi^*_a-\pi^*_b)\t I_2|$, it follows:
    \begin{align*}
       \alpha\mu \|\pi^*_a-\pi^*_b\|_2^2 \leq|(\pi^*_a-\pi^*_b)\t I_2| &\leq  \|\nabla_\pi f(\pi^*_b; \theta_a) - \nabla_\pi f(\pi^*_b; \theta_b)\|\|\pi_1-\pi_2\|_2\\
        &\leq  C d_\theta (\theta_a, \theta_b)\|\pi^*_a-\pi^*_b\|_2,
    \end{align*}
    and, consequently, from the two extreme inequalities, we have, $\alpha\mu \|\pi^*_a-\pi^*_b\|_2^2 \leq  C d_\theta (\theta_a, \theta_b).$
    Using the relation $\dtv{\pi^*_a, \pi^*_b}\leq \frac{\sqrt{n}}{2} \|\pi^*_a-\pi^*_b\|_2 $, where $n$ is the dimension of $\pi$ (in order case, $|\U|$) and defining $C={\bar{C}\sqrt{n}}/{2}$, we have the result.

\end{proof}

\begin{lemma}\label{lmm:opt-lipschitz}
    The operator $\B^{\textrm{RQE}}_\opt(\S_\mathbb{M})$ is Lipschitz with constant $\propto {1}/{\alpha}$.
\end{lemma}
\begin{proof}
    Consider two sets of mean-field flows $\S^a_\mathbb{M}$ and $\S^b_\mathbb{M}$, with corresponding optimal policies $\pi^*_a$ and $\pi^*_b$.
    At $t=T-1$, for any state $x$, we have 
    \begin{align*}
     c_{T-1}^{(\pi_{T-1}, \pi^*_{a,-{T-1}}), ~\alpha}(x; \S^a_\mathbb{M})  = \frac{1}{\tau}\textrm{LSE} \Big(\log w_k-\tau \sum_{u\in\U}\pi_{T-1}(u|x)r_{T-1}(x, u, \mu^{a, k}_{T-1}) \Big) + \alpha\nu(\pi_{T-1}(\cdot|x)),
    \end{align*}
    where the notation follows from the proof of Theorem~\ref{thm:MF-RQE-finite-population}.
    Since the LSE term is convex in $\pi$ and $\nu(\pi_{T-1}(\cdot|x))$ is $m$-strongly convex, we have that $c_{T-1}^{(\pi_{T-1}, \pi^*_{a,-{T-1}}), ~\alpha}(x; \S^a_\mathbb{M})$ is $(\alpha m)$-strongly convex.
    This implies, $\pi^*_{a, T-1}(\cdot | x)$ is the unique optimizer of the cost function.
    Define $\beta_k =\frac{w_k \exp(-\tau\pi_t\t(\cdot | x)Q^\pi_{\mu^k, t}(x, \cdot))}{\sum_j w_j \exp(-\tau\pi_t\t(\cdot | x)Q^\pi_{\mu^j, t}(x, \cdot))}$.
    Then,
    \begin{align*}
        \nabla_{\pi_{T-1}}c_{T-1}^{(\pi_{T-1}, \pi^*_{a,-{T-1}}), ~\alpha}(x; \S^a_\mathbb{M}) = \sum_k\beta_k r_{T-1}(x, \cdot, \mu^{a, k}_{T-1}) + \alpha\nabla_\pi \nu(\pi_{T-1}).
    \end{align*}
    From Assumption~\ref{assmpt:lipschitiz-model}, 
    \begin{align*}
        &\|\nabla_{\pi_{T-1}}c_{T-1}^{(\pi_{T-1}, \pi^*_{a,-{T-1}}), ~\alpha}(x; \S^a_\mathbb{M}) -\nabla_{\pi_{T-1}}c_{T-1}^{(\pi_{T-1}, \pi^*_{b,-{T-1}}), ~\alpha}(x; \S^b_\mathbb{M}) \|\\ &\leq \sum_k\beta_k \|r_{T-1}(x, \cdot, \mu^{a, k}_{T-1}) -r_{T-1}(x, \cdot, \mu^{b, k}_{T-1})\|\\
        &\leq \sum_k \beta_k L_r |\U|\dtv{\mu^{a, k}_{T-1}, \mu^{b, k}_{T-1}}\\
        &\leq C_{T-1}d_\S(\S^a_\mathbb{M}, \S^b_\mathbb{M}),
    \end{align*}
    where $\sum_k\beta_k = 1$.
    Now using the above equation with Lemma~\ref{lmm:strong-convexity}, where $\theta_a = (\S^a_{\mathbb{M}}, \pi^*_{a,-{T-1}})$ and $\theta_b = (\S^b_{\mathbb{M}}, \pi^*_{b,-{T-1}})$,
    \begin{align*}
         \dtv{\pi^*_{a, T-1}(\cdot|x), \pi^*_{b, T-1}(\cdot|x)}\leq \frac{C_{T-1}}{\alpha m}d_\S(\S^a_\mathbb{M}, \S^b_\mathbb{M}),
    \end{align*}
    for all states $x\in\X$.
    Consequently, 
    \begin{align*}
        d_\Pi(\pi^*_{a, T-1}, \pi^*_{b, T-1}) &\triangleq \max_{x\in\X} \dtv{\pi^*_{a, T-1}(\cdot|x), \pi^*_{b, T-1}(\cdot|x)}\\&\leq \frac{C_{T-1}}{\alpha\mu}d_\S(\S^a_\mathbb{M}, \S^b_\mathbb{M}).
    \end{align*}
    
    For $t=T-2$, we similarly have,
    \begin{align*}
        &\|\nabla_{\pi_{T-2}}c_{T-1}^{(\pi_{T-2}, \pi^*_{a,-{T-2}}), ~\alpha}(x; \S^a_\mathbb{M}) -\nabla_{\pi_{T-2}}c_{T-2}^{(\pi_{T-2}, \pi^*_{b,-{T-2}}), ~\alpha}(x; \S^b_\mathbb{M}) \|\\ &\leq \sum_k\beta_k \|Q^{\pi^*_a}_{\mu^{a,k}, T-2}(x, \cdot))- Q^{\pi^*_b}_{\mu^{b,k}, T-2}(x, \cdot))\|\\
        &\leq \bar{C}_{T-1}d_\S(\S^a_\mathbb{M}, \S^b_\mathbb{M}) + |B_r|d_\Pi(\pi^*_{a, T-1}, \pi^*_{b, T-1})\\
        &\leq \bar{C}_{T-1}d_\S(\S^a_\mathbb{M}, \S^b_\mathbb{M}) + \frac{C_{T-1}}{\alpha m}d_\S(\S^a_\mathbb{M}, \S^b_\mathbb{M})\\
        &\triangleq C_{T-2}d_\S(\S^a_\mathbb{M}, \S^b_\mathbb{M}),
    \end{align*}
    where we have used continuity of the Q-function derived in Lemma~\ref{lmm:continuity-Q}.
    Once again, following Lemma 6, we can repeat the previous steps and show that
    \begin{align*}
        d_\Pi(\pi^*_{a, T-2}, \pi^*_{b, T-2}) &\triangleq \max_{x\in\X} \dtv{\pi^*_{a, T-2}(\cdot|x), \pi^*_{b, T-2}(\cdot|x)}\\&\leq \frac{C_{T-2}}{\alpha m}d_\S(\S^a_\mathbb{M}, \S^b_\mathbb{M}).
    \end{align*}

    One can iteratively follow these steps to show that the result holds for all $t\in\{0, \ldots, T-1\}$ and by defining $C = \max_{t
< T} C_t$,
we have,
\begin{align*}
        d_\Pi(\pi^*_{a, t}, \pi^*_{b, t}) &\leq \frac{C}{\alpha m}d_\S(\S^a_\mathbb{M}, \S^b_\mathbb{M}).
    \end{align*}
    This implies $d_\Pi(\pi^*_a, \pi^*_b) = d_\Pi(\B^{\textrm{RQE}}_\opt(\S^a_\mathbb{M}), \B^{\textrm{RQE}}_\opt(\S^b_\mathbb{M})) = \max_{t<T} d_\Pi(\pi^*_{a, t}, \pi^*_{b, t})\leq \frac{C}{\alpha m}d_\S(\S^a_\mathbb{M}, \S^b_\mathbb{M})$.
\end{proof}

\begin{proof}
We prove the statement by showing that $\Phi$ is a contraction mapping and then apply Banach's fixed point theorem.
Let $\pi, \pi'\in\Pi$ be arbitrary, then
\begin{align*}
    d_\Pi(\Phi(\pi), \Phi(\pi')) &= d_\Pi\Big(\B^{\textrm{RQE}}_\opt\circ\B^{\textrm{RQE}}_\prop(\pi), \B^{\textrm{RQE}}_\opt\circ\B^{\textrm{RQE}}_\prop(\pi')\Big)\\
    &\leq \frac{C}{\alpha m}d_\S(\B^{\textrm{RQE}}_\prop(\pi), \B^{\textrm{RQE}}_\prop(\pi'))\\
    &\leq \frac{C K_\prop}{\alpha m}d_\Pi(\pi, \pi')\\
    &= \kappa d_\Pi(\pi, \pi'),
\end{align*}
where we use, in order, Lemmas~\ref{lmm:opt-lipschitz} and~\ref{lmm:prop-lipschitz}.
Thus for sufficiently large $\alpha$, $\frac{C K_\prop}{\alpha m} =\kappa < 1$, and $\Phi$ is a contraction mapping.
Moreover, the existence of a fixed point follows from Proposition~\ref{prop:existence-MF-RQE}.
Then, by Banach's fixed point theorem, FPI converges.
\end{proof}

\section{Proof of Theorem~\ref{thm:Fictitious-Play-convergence}}
\begin{proof}
   We view the problem as the evolution of a discrete time dynamical system.
   Specifically, define states $Y_1$ and $Y_2$ that follow,
   \begin{align*}
       Y_1(k+1) &= \Phi(Y_2(k)), \\
       Y_2(k+1) &= \beta Y_2(k) + (1-\beta) Y_1(k).
   \end{align*}
   Solving for the equilibrium conditions, we get $Y^*_1 = Y^*_2$, which is exactly the definition of the MF-RQE.
   Thus, all that is left to show is that the equilibrium is attained as $k\to\infty$.
   Define $Y(k) = [Y_1(k), Y_2(k)]$ and the associated norm $\|Y\| \triangleq \max\{\|Y_1\|_\Pi, \|Y_2\|_\Pi\}$.
   Now consider $Y(k) = [Y_1(k), Y_2(k)]$ and $\bar Y(k) = [\bar Y_1(k), \bar Y_2(k)]$.
   Then $\|Y(k+1) - \bar{Y}(k+1)\| = \max\{\|Y_1(k+1) -\bar Y_1(k+1)\|_\Pi, \|Y_2(k+1) -\bar Y_2(k+1)\|_\Pi\} $.
   We have\footnote{%
   Recall that $ \|Y_1(k+1) -\bar Y_1(k+1)\|_\Pi = d_\Pi(Y_1(k+1),\bar Y_1(k+1))$.} that
   \begin{align*}
       \|Y_1(k+1) -\bar Y_1(k+1)\|_\Pi &= \|\Phi(Y_2(k))-\Phi(\bar Y_2(k))\|_\Pi\\
       &<\kappa\|Y_2(k)-\bar Y_2(k)\|_\Pi\\
       &\leq \kappa\|Y(k) -\bar Y(k)\|,
   \end{align*}
   where the second inequality holds for large enough $\alpha$, as shown in Theorem~\ref{thm:FPI-convergence}.
   We also have, 
   \begin{align*}
       \|Y_2(k+1) -\bar Y_2(k+1)\|_\Pi &= \|\beta Y_2(k) + (1-\beta)Y_1(k)-\beta\bar Y_2(k) - (1-\beta)\bar Y_1(k)\|_\Pi\\
       &\leq \beta\|Y_2(k)-\bar Y_2(k)\|_\Pi + (1-\beta)\|Y_1(k)-\bar Y_1(k)\|_\Pi\\
       &\leq \max\{ \|Y_2(k)-\bar Y_2(k)\|_\Pi , \|Y_1(k)-\bar Y_1(k)\|_\Pi\}\\
       &=\|Y(k) -\bar Y(k)\|.
   \end{align*}
   Thus, $\|Y(k+1) - \bar{Y}(k+1)\| = \max\{\|Y_1(k+1) -\bar Y_1(k+1)\|_\Pi, \|Y_2(k+1) -\bar Y_2(k+1)\|_\Pi\} \leq \|Y(k) -\bar Y(k)\|$.

   Following the similar procedure for $k+2$, we can show
   \begin{align*}
       \|Y_1(k+2) -\bar Y_1(k+2)\|_\Pi &= \|\Phi(Y_2(k+1))-\Phi(\bar Y_2(k+1))\|_\Pi\\
       &<\kappa\|Y_2(k+1)-\bar Y_2(k+1)\|_\Pi\\
       &< \kappa^2\|Y(k) -\bar Y(k)\|,
   \end{align*}
   and 
   \begin{align*}
       \|Y_2(k+2) -\bar Y_2(k+2)\|_\Pi &= \|\beta Y_2(k+1) + (1-\beta)Y_1(k+1)-\beta\bar Y_2(k+1) - (1-\beta)\bar Y_1(k+1)\|_\Pi\\
       &\leq \beta\|Y_2(k+1)-\bar Y_2(k+1)\|_\Pi + (1-\beta)\|Y_1(k+1)-\bar Y_1(k+1)\|_\Pi\\
       &\leq \beta \|Y(k) -\bar Y(k)\| + (1-\beta) \kappa\|Y(k) -\bar Y(k)\| \\
       &<\|Y(k) -\bar Y(k)\|,
   \end{align*}
   since $\kappa<1$.
   It follows that $\|Y(k+2) - \bar Y(k+2)\| < \|Y(k) - \bar Y(k)\|$. 
   For $\bar{Y} = [Y^*_1, Y_2^*]$, we have established convergence the equilibrium as $k\to\infty$.
   Hence, for $Y_2 = \bar\pi$ and $Y_1=\pi$, and the result follows.
\end{proof}

\end{document}